\RequirePackage{fix-cm}
\documentclass[12pt, leqno]{article}

\usepackage{amsmath,amsthm,amssymb,graphicx, graphics, epsfig}
\usepackage{bm}

\usepackage{manfnt}
\usepackage{palatino}
\usepackage[sc]{mathpazo}
\usepackage{eulervm}
\usepackage{fullpage}
\newcommand{\comment}[1]{}

\newtheorem{thm}{Theorem}
\newtheorem{lem}{Lemma}
\newtheorem{cor}{Corollary}
\newtheorem{prop}{Proposition}
 
\theoremstyle{definition}
\newtheorem*{defi}{Definition}
\newtheorem*{rem}{Remark}

\renewcommand{\proof}{\medskip\noindent\textit{Proof. }}
 
\newcommand{\phys}[1]{\footnote{#1}}
 
\newcommand{\q}[1]{{``#1''}}
\newcommand{\br}[1]{\left(#1\right)}
\newcommand{\brr}[1]{\left[#1\right]}
\newcommand{\bra}[1]{\left\langle #1 \right\rangle}
\newcommand{\set}[1]{\left\{#1\right\}}
\newcommand{\norm}[1]{\left\|#1\right\|}
\newcommand{\case}[1]{\begin{cases}#1\end{cases}}
\newcommand{\inv}{^{-1}}
\newcommand{\sumj}{\sum_{j=1}^m}
 
\newcommand{\restr}[2]{\left. #1 \right|_{#2}}
\newcommand{\EY}{\restr{E}{\partial X}}
\newcommand{\tin}{_{t\in [0,1]}}
\newcommand{\XgN}{_{X,\,g,\,N}} 
\newcommand{\XN}{_{X,\,N}} 
 
\newcommand{\matr}[1]{
\begin{pmatrix}
#1
\end{pmatrix}}

\newcommand{\sys}[1]
{\left\{
\begin{array}
#1
\end{array}
\right.}

\newcommand{\End}{\mathsf B}
\newcommand{\Herm}{\mathsf H}

\newcommand{\HF}{\mathsf{HF}}
\DeclareMathOperator{\GL}{GL}

\newcommand{\R}{\mathbb R}

\newcommand{\Z}{\mathbb Z}
\newcommand{\CC}{\mathbb C}

\newcommand{\B}{\mathcal B}
\newcommand{\DD}{\mathcal D}
\newcommand{\DB}{\DD\times\B}

\newcommand{\Dir}{\mathbb D}
\newcommand{\DG}{D}
\newcommand{\Dt}{\DG}
\newcommand{\DtB}{(\Dt,B)}

\DeclareMathOperator{\Ker}{Ker}

\DeclareMathOperator{\Hom}{Hom}
\DeclareMathOperator{\Diff}{Diff}
\DeclareMathOperator{\const}{const}

\DeclareMathOperator{\dom}{domain}
\DeclareMathOperator{\dvol}{dvol}
\DeclareMathOperator{\tg}{tan}
\DeclareMathOperator{\ctg}{cot}
 
\DeclareMathOperator{\spf}{\mathsf{sf}}
\renewcommand{\sp}[1]{\spf\,(#1)}
 
\renewcommand{\mod}{\,\mathrm{mod}\,}
\renewcommand{\d}{d}
\newcommand{\x}{x}
\newcommand{\n}{n}
\newcommand{\sqg}{\sqrt{g}}

\newcommand{\p}{\partial}
\newcommand{\px}[1]{\partial_1 #1}
\newcommand{\py}[1]{\partial_2 #1}

\newcommand{\new}{_{\mathrm{new}}}
\newcommand{\cl}{\mathrm{c}}

\renewcommand{\leq}{\leqslant} 
\renewcommand{\geq}{\geqslant}

\newcommand{\brop}[1]{\br{#1}} 

\newcommand{\addcontsec}[1]{\addcontentsline{toc}{section}{\protect\numberline {#1}}}

\begin{document}

\title{The spectral flow for Dirac operators on compact planar domains with local boundary conditions}
\author{Marina\,Prokhorova}
\date{}
 
\maketitle

\let\oldnumberline=\numberline
\renewcommand{\numberline}{\vspace{-4mm}\oldnumberline}
\addtocontents{toc}{\protect\renewcommand{\bfseries}{}}

\comment{
\begin{abstract}
Let $D_t$, $0\leqslant t \leqslant 1$ be an arbitrary 1-parameter family of Dirac type operators on a two-dimensional disk with $m-1$ holes.
Suppose that all operators $D_t$ have the same symbol, and that $D_1$ is conjugate to $D_0$ by a scalar gauge transformation.
Suppose that all operators $D_t$ are considered with the same elliptic local boundary condition.
Our main result is a computation of the spectral flow for such a family of operators.
The answer is obtained up to multiplication by an integer constant depending only on the number of the holes in the disk.
This constant is calculated explicitly for the case of the annulus ($m=2$).
\end{abstract}
}


\tableofcontents
\bigskip

\sloppy

\section{Introduction}

This paper deals with Dirac type operators on compact planar domains.
We consider such operators with 
self-adjoint locally elliptic local boundary conditions\footnote{In particular, 
boundary conditions defined by general pseudo-differential operators are not allowed.}.
The paper is focused not on individual operators, 
but on paths in the space of such operators.
We consider only paths connecting two operators conjugate by a scalar gauge transformation 
(so, they are loops up to a scalar gauge transformation).
Such paths have a well known invariant, the spectral flow, 
which counts with signs the number of eigenvalues passing through zero from the start of the path to its end 
(the eigenvalues passing from negative values to positive one are counted with the plus sign, 
and egenvalues passing in the other direction are counted with the minus sign). 
The paper is devoted to the problem of computation of the spectral flow
in the situation when all the operators along the path have the same symbol and the same boundary condition.

Because these results are potentially useful for the physics of condensed matter,
the author attempted to avoid advanced mathematical terminology 
and to explain the results and the ideas behind their proofs 
in a way accesible to non-mathematicians.
By the same reason, the author present the case of Dirac operators 
(Theorems \ref{thm:Dirac} and \ref{thm:2NDirac}) 
before dealing with the more general case of Dirac type operators on domains equipped with an arbitrary metric.
Note that physicists sometimes use more general boundary conditions than the ones considered in this paper. 
For example, the so-called armchair boundary conditions for graphen are of the type considered in this paper,
but the zigzag boundary conditions for graphen are not.
While it is not completely obvious, boundary conditions 
considered in this paper are just another forms of locally elliptic boundary conditions used in physics,
as explained in Section \ref{sec:phys}.
Besides, we explain in Section \ref{sec:phys} how the spectral flow can be computed in terms of 
general boundary problem for two-dimensional spinors considered by Akhmerov and Beenakker in \cite{AhmBeen}.
As illustration we show that the spectral flow vanishes if boundary condition does not break the time reversal symmetry. 

We start with the following situation.
Let $X$ be a compact planar domain bounded by $m$ smooth curves (topologically it is a disk with $m-1$ holes).
Our operators act on spinor (i.~e. spinor-valued) functions,
which we identify with column vectors of two complex-valued functions: 
$$
u=\matr{ u^+ \\ u^- }, \; u^{\pm}\colon X\to\CC.
$$
A Dirac operator acting on spinor functions has the form
$$
\DG = \Dir + \matr{ 0 & \bar{q}(\x) \\ q (\x) & 0}, \; 
\Dir = -i\matr{ 0 & \px{}-i\py{} \\ \px{}+i\py{} & 0},
$$
where $q$ is a smooth function from $X$ to $\CC$,
$x=(x^1,x^2)\in X$, $\p_i=\p/{\p x^i}$.
Our focus is on $1$-parameter families $D_t$ of such operators parametrized by $t\in [0,1]$.
In such a family the first term $\Dir$ involving derivatives is always the same, but
in the second term the function $q$ is allowed to continuously change with $t$, 
i.~e. $q=q_t$, where $t\in [0,1]$.
In agreement with the above, we assume that $D_1 = \mu  D_0 \mu \inv$ 
for some smooth scalar gauge transformation $\mu \colon X \to U(1)$.

\comment{
Let $D_t$ be a 1-parameter family of two-dimensional Dirac operators acting on a spinor (i.~e. spinor-valued) function $u$:
$$D_t= -i\matr{ 0 & \px{}-i\py{} \\ \px{}+i\py{} & 0} + \matr{ 0 & \overline{q_t}(\x) \\ q_t(\x) & 0}, \; 
u=\matr{ u^+ \\ u^- }, \; u^{\pm}\colon X\to\CC,$$
where $q_t(\x)$ is a smooth function from $X$ to $\CC$ continuously depending on $t\in [0,1]$ 
such that $D_1 = \mu  D_0 \mu \inv$ for some smooth gauge transformation $\mu \colon X \to U(1)$.
}

All operators $\Dt_t$ are considered with the same boundary condition of the form $i(\n_1 + i\n_2)u^+ =  B(\x) u^-$,
where $\n=(\n_1,\n_2)$ is the outward conormal to the boundary,
and $B$ is a real-valued smooth function on the boundary of $X$ without zeros. 
Our first main result, Theorem \ref{thm:Dirac}, 
asserts that the spectral flow of such a family of operators is equal to $c_m \sumj b_j \mu_j$,
where $c_m$ is an integer constant depending on $m$ only,
$\mu_j$ is the degree of the restriction of $\mu$ to $j$-th connected boundary component,
$b_j=1$ if $B$ is negative on the $j$-th boundary component, and equal to $0$ otherwise.

After considering this most special and very important situation, 
we turn our attention to the situation of Dirac operators acting on $N$-dimensional spinor functions
$$
u=\matr{ u^+ \\ u^- }, \; u^{\pm}\colon X\to\CC^N,
$$
where, as before, $X$ is a compact planar domain bounded by $m$ smooth curves.
A Dirac operator acting on $N$-dimensional spinor functions has the form
$$
\DG = \Dir+Q(\x), \; \Dir = -i \br{ \sigma_1 \px{} + \sigma_2 \py{} },
$$
where
$$
\sigma_1 =
\matr{
    0 & I_N \\
    I_N & 0
}, \;
\sigma_2 =
\matr{
    0 & -iI_N \\
    iI_N & 0
},
$$
$I_N$ is $N\times N$ unit matrix, 
and $Q(\x)$ is complex self-adjoint $2N\times 2N$ matrix smoothly dependent on $\x\in X$.
Again, our focus is on $1$-parameter families $D_t$ of such operators parametrized by $t\in [0,1]$.
In such a family the first term $\Dir$ involving derivatives is always the same, but
in the second term the matrix $Q$ is allowed to continuously change with $t$, 
i.~e. $Q=Q_t$, where $t\in [0,1]$.
We assume that $\DG_1 = \mu  \DG_0 \mu \inv$ 
for some smooth scalar gauge transformation $\mu \colon X \to U(1)$, 
where $U(1)$ is considered as the subgroup of $U(2N)$ consisting of the diagonal matrices with equal diagonal elements. 

All operators $\DG_t$ are considered with the same boundary condition $i(\n_1 + i\n_2)u^+ =  B(\x) u^-$,
where $B$ is a smooth map from the boundary to the space 
of complex self-adjoint invertible $N\times N$ matrices.
Note that a local boundary condition is locally elliptic if and only if 
it can be written in such a form with $B(\x)$ invertible at any $\x$;
this boundary condition is self-adjoint if and only if $B(\x)$ is self-adjoint at any $\x$.

Our second main result, Theorem \ref{thm:2NDirac}, asserts that the spectral flow  
of such a family of operators is equal to $c_m \sumj b_j \mu_j$,
where $c_m$ is the same constant as in Theorem \ref{thm:Dirac} 
(in particular, $c_m$ does not depend on the dimension $N$),
$\mu_j$ is the degree of the restriction of $\mu $ to $j$-th boundary component 
(this restriction gives us the map from the circle to the circle because $\mu$ is a scalar gauge transformation), 
and $b_j$ is the number of negative eigenvalues of $B(\x)$ (counting with multiplicities) on the $j$-th boundary component.

Theorem \ref{thm:gDirac} extends Theorem \ref{thm:2NDirac} to a still more general class of operators. 
These Dirac type operators involve in their definition an arbitrary (not necessarily flat) metric on $X$,
and have the principal symbol defined by a Clifford multiplication which does not necessarily agree with this metric.
While considering arbitrary metric is important for some physical applications, 
considering Clifford multiplication which does not agree with the metric on $X$ does not seem neseccary.
Nevertheless, we take care of this more general case because the proofs of our results crucially depend on its consideration.
Moreover, we cannot prove Theorems \ref{thm:Dirac} and \ref{thm:2NDirac} without proving Theorem \ref{thm:gDirac} first.

Notice that \textit{scalar} gauge transformations $\mu \colon X \to U(1)$ leave invariant every local boundary condition, 
as well as the first term $\Dir$ of the operator $\Dir+Q(\x)$: $\mu (\Dir+Q) \mu \inv = \Dir+Q'$ for some function $Q'(\x)$.
So any operator $D_0$ can be connected with the conjugate operator $\DG_1 = \mu  \DG_0 \mu \inv$ by the path $(\Dir+Q_t(\x))$ with fixed boundary condition. 
On the contrary, non-scalar gauge transformations $\mu \colon X \to U(2N)$ do not have such properties, 
so the problem of computation of the spectral flow can not be stated in such a form as described above.
If we allow general non-scalar gauge transformations, 
then we have to allow paths of operators $(D_t)$ and of boundary conditions $(B_t)$ 
with $B_t$ and symbol of $D_t$ being dependent on $t$. 
Some results about this more general case are outlined in Section \ref{elliptic}.

Note that in this paper the spectral flow is computed only up to multiplication 
by an integer constant $c_m$ depending only on $m$. 
For a disk with one hole ($m=2$) the eigenvectors and hence the spectral flow are calculated explicitly in a special case; 
this is sufficient to determine $c_2$; it turns out that $c_2=1$ (see Theorem \ref{thm:c2}).
For the case $m>2$ this method fails because Fourier transform gives no help here.
Nevertheless, the author expects that $c_m=1$ for all $m$; some reasons in favour of this conjecture are provided after proving that $c_2=1$.

\comment{
The exposition is arranged as follows.
In Section \ref{seq:sf} we recall the definition of the spectral flow.
The main result of the paper is stated in Section \ref{sec:gDO} (Theorem \ref{thm:gDirac}),
while Sections \ref{sec:DO} (Theorem \ref{thm:Dirac}) and
\ref{sec:2NDO} (Theorem \ref{thm:2NDirac}) contain the simpler
versions of this result. In Section \ref{sec:annulus} we
calculate the above mentioned multiplicative constant for a disk with one hole,
and in Section \ref{sec:manyholes} discuss this constant for disks with several holes.
In Section \ref{elliptic} we discuss further generalizations of these results.
Part \ref{part:2} contains the proof of the main result (Theorem \ref{thm:gDirac}). 
}

\part{The spectral flow for Dirac operators}\label{part:1}

\section{The spectral flow}\label{seq:sf}

Let $H$ be a complex separable Hilbert space,
$\br{A_t}$, $t \in [0,1]$ be a continuous 1-parameter family of bounded self-adjoint (or, what is the same, Hermitian) Fredholm operators in $H$.
Near zero every $A_t$ has discrete real spectrum, which changes continuously with the variation of $t$.
Hence one can count the net number of eigenvalues of $A_t$ passing through zero in positive direction 
as $t$ runs from $0$ to $1$,
that is, the difference between the numbers of eigenvalues (counting multiplicities)
crossing zero in positive and negative directions.
This net number is called the spectral flow $\sp{A_t}$.
The description of this notion can be found in \cite{APS-76, BBW-93}.

The case when $A_0$ or $A_1$ has zero eigenvalue requires some agreement on the counting procedure;
we use the following convention: take a small $\varepsilon>0$ such that $A_0$, $A_1$
have no eigenvalues in the interval $[-\varepsilon, 0)$, and
define the spectral flow as the net number of eigenvalues of
$A_t+\varepsilon I$ which pass through zero.

Let now $\br{A_t}$ be an 1-parameter family of (not necessarily bounded) self-adjoint Fredholm operators in $H$.
For example, it can be a family of symmetric elliptic differential operators $A_t$ acting on sections
of Hermitian bundle $E$ over closed (that is, compact without boundary) manifold $X$.
The definition of the spectral flow can be adjusted to this case, 
though more accurate consideration is needed,
particularly due to the presence of various natural topologies on the space of such operators \cite{BLP-04, L-04, BLP-02}.

When a manifold has non-empty boundary, we have to consider the family $\br{A_t, B_t}$,
where $A_t$ is a formally self-adjoint differential operator,
and $B_t$ is a ``good'' (self-adjoint elliptic) boundary condition for $A_t$ at any $t$.
One can see the notion of self-adjoint elliptic boundary value problem for operators of Dirac type in \cite{BBW-93, BL-01},
and for general first order elliptic operators in \cite{BBLZ-09}.

Such differential operator $A_t$ with boundary condition $B_t$ defines the unbounded self-adjoint Fredholm operator on $L^2(X,E)$,
which has unbounded discrete real spectrum.
Intuitively, the spectrum of $\br{A_t, B_t}$ changes continuously with the variation of $t$,
so the definition of the spectral flow works in this case as well \cite{BLP-04,BLP-02}.
However, the proof that the definition and the standard properties
of the spectral flow are correct is considerably more difficult in this case.
The crucial ingredient is the continuity (in $t$) of the family $\br{A_t,B_t}$ in the space of
unbounded self-adjoint Fredholm operators on $L^2(X,E)$ with an appropriate metric.
This was proved in \cite{BBLZ-09} (see \cite{BBLZ-09}, Theorem 7.16).
This continuity property allows one to use the theory developed in \cite{BLP-04, L-04} in full force.
Our proof of Theorem \ref{thm:gDirac} (see Part \ref{part:2}) crucially depends on this theory,
and, in particular, on Theorem 7.16 from \cite{BBLZ-09}.
The results of this theory needed for the proof of Theorem \ref{thm:gDirac}
are isolated in Section \ref{sec:P0-P4} as properties (P0-P4).

Note that if the spectra of $(A_0,B_0)$ and $(A_1,B_1)$ are the same (isospectral operators),
which is  the case in this paper, then there is another way to define the spectral flow of $(A_t,B_t)$.
The set
$$
\set{ (t,\lambda) \colon \lambda \mbox{ is the eigenvalue of } (A_t,B_t)}
$$
can be uniquely represented as the union of the graphs of
functions $\lambda_i(t)$ such that $\lambda_i(t) \leq
\lambda_j(t)$ for $i \leq j$. These functions give us the bijection
(one-to-one correspondence) of the spectrum of $(A_0,B_0)$ to the
spectrum of $(A_1,B_1)$. If these spectra coincide as subsets of $\R$
then this correspondence gives us the shift of the spectrum on the
integer number of positions. This number is the spectral flow of
$(A_t,B_t)$. It is worth to note that for the isospectral
case one can replace the level $\lambda=0$ by any real number, and
the difference between eigenvalues crossing the level in positive
and negative directions will be the same \cite{APS-76}.

\section{Dirac operators: the simplest case}\label{sec:DO}

Suppose $X$ is a compact planar domain bounded by $m$ smooth curves 
(topologically it is a disk with $m-1$ holes). 
We will use the notations $x=(x^1,x^2)\in X$, $\p_i=\p/{\p x^i}$.

Let us consider the Dirac operator on $X$
\begin{equation}\label{eq:Dirac}
\Dir= -i
\matr{
    0 & \px{}-i\py{} \\
    \px{}+i\py{} & 0
},
\end{equation}
acting on a spinor function $u \colon X \to \CC^2$, $u=\matr{ u^+ \\u^- }$.

A Dirac operator with non-zero vector potential has the form 
\begin{equation*}
\DG = \Dir + Q(\x), \mbox{ where }
Q(\x)=\matr
{   0 & \bar{q}(\x) \\
    q(\x)  & 0
},
\end{equation*}
$q$ is a smooth function from $X$ to $\CC$.

Let $\mu  \colon X \to U(1)$ be a gauge transformation;
we suppose that $\mu (\x) \in \CC$, $|\mu (\x)|\equiv 1$ for $\x\in X$.
Let us take a Dirac operator $\Dt_0 = \Dir+Q_0(\x)$ and connect it with the conjugate operator
\[
\Dt_1 = \mu \Dt_0 \mu \inv = \Dir+Q_0 +
    \matr{
    0 &  i \mu\inv\br{ \px{\mu } - i\py{\mu } } \\
    i \mu \inv\br{ \px{\mu } + i\py{\mu } }  & 0
}
\]
by an one-parameter family of Dirac operators
\begin{equation}\label{eq:Qt}
\Dt_t = \Dir +Q_t, \mbox{ where }
Q_t(\x)=\matr
{   0 & \overline{q}_t(\x) \\
    q_t(\x)  & 0
},
\end{equation}
$q_t$ is a smooth function from $X$ to $\CC$ continuously
depending on $t$, $t\in [0,1]$, 
$$q_1-q_0 = i\mu\inv\br{ \px{\mu } + i\py{\mu } }.$$

A self-adjoint locally elliptic\footnote{Another name for \q{locally elliptic boundary condition} is \q{Sapiro-Lopatinskii boundary condition}} 
local boundary condition for $\Dt_t$ has the form
\begin{equation}\label{eq:bc}
  in(x)u^+ = B(\x) u^- \mbox{ on } \p X,
\end{equation}
where $B \colon \partial X \to \R \setminus\set{0}$ is a smooth function defining the boundary condition,
$\n=\br{\n_1,\n_2}$ is the outward conormal to the boundary $\p X$ of $X$ at point $\x$, 
and we identify $n$ with the complex number $\n_1+i\n_2$ in \eqref{eq:bc}.

Note that $\n_1$, $\n_2$ coincide with the components of the outward normal to $\p X$
for the case of Euclidean metric considered both here and in the next section.
In Section \ref{sec:gDO} we consider a more general case of arbitrary metric on $X$, 
and the distinction between a normal and a conormal becomes essential there.

\begin{rem}\label{rem:connection}
$\Dir+Q(\x)$ is the Dirac operator on the trivial 2-dimensional complex vector bundle over $X$
with compatible unitary connection defined by the function $q(\x)$.
So the change of $q_t$ with $t$ is equivalent to the change of the connection.
\end{rem}

Boundary condition \eqref{eq:bc} is gauge invariant with respect to the conjugation by $\mu$,
while $\Dt_0$ and $\Dt_1$ are conjugate by $\mu$.
So the operators $\Dt_0$, $\Dt_1$  with the same boundary condition \eqref{eq:bc} are isospectral,
and the spectral flow of the family $\Dt_t$ gives us the shift of the spectrum of $\Dt_t$ when $t$ runs from $0$ to $1$.

\begin{thm}\label{thm:Dirac}
The spectral flow of the family $(\Dt_t)$, $t\in [0,1]$, with boundary condition \eqref{eq:bc} is equal to
\[
c_m \sumj b_j \mu_j,
\]
where
$c_m$ is the integer constant depending on $m$ only,
$\mu_j$ is the degree of the restriction of $\mu $ to $\p X_j$,
\[
b_j =
\sys{ {cc}
1,  \mbox{ if } B<0  \mbox{ on } \p X_j \\ 0, \mbox{ if } B>0  \mbox{ on } \p X_j
}
\]
Here $\p X_j$ are the connected components of the boundary of $X$,
equipped with the orientation in such a way that the pair
(outward normal to $\p X_j$, positive tangent vector to $\p X_j$) has the positive orientation on the plane $(x^1,x^2)$.
\end{thm}

\begin{figure}[tbh]
\begin{center}
\includegraphics[width=0.28\textwidth]{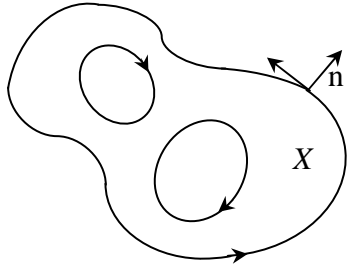}
\caption{The case of two holes}
\label{fig:X}
\end{center}
\end{figure}

Note that since $B \neq 0$, it has definite sign at each boundary component $\p X_j$, so the constants $b_j$ are correctly defined.
The restriction of $\mu $ to $j$-th connected component of $\p X$ gives us the map from the circle $\p X_j$ to the circle $U(1)$; $\mu_j$ is the degree of this map. 

\comment{The degree of this map is the number of revolutions of $\mu(\x)$ when $\x$ runs $\p X_j$ in the positive direction.
Here the revolutions in counterclockwise direction counts with the plus sign, and the revolutions in clockwise direction counts with the minus sign.}

This theorem follows from more general result which we formulate below. 
The generalization goes in two directions:
(1) we allow arbitrary dimension of unknown complex functions
$u^-$, $u^+$, (2) we replace Dirac operator by operators of more
general form.
The value of $c_2$ is computed in Section \ref{sec:annulus}.

\begin{rem}\label{rem:Berry}
Boundary condition \eqref{eq:bc} coincides with the boundary
condition of Berry and Mondragon for the \q{neutrino billiard}
\cite{Berry} up to replacement of $B$ by $B\inv$.
In physical terms, one-parameter family of Dirac operators \eqref{eq:Qt} describes the situation of 
continuously varying magnetic field so that the following two conditions are fullfiled:
\vspace{-5pt}
\begin{enumerate}
\parsep=0pt 
\itemsep=\parsep
	\item[(1)] magnetic field at $t=1$ coincides with magnetic field at $t=0$ all over the interior of $X$,
	\item[(2)] the fluxes through $j$-th hole at $t=1$ and at $t=0$ differ by integer number $\mu_j$ in the units of the flux quantum.
\end{enumerate}\vspace{-5pt}
Let $j=m$ corresponds to the outer boundary component and $j=1,\ldots, m-1$ enumerate the holes.
Considering that $\mu_m = -\sum_{j=1}^{m-1} \mu_j$, 
we can reformulate Theorem \ref{thm:Dirac} as following:
the spectral flow of the operators family \eqref{eq:Qt} with boundary condition \eqref{eq:bc} is equal to
\[
c_m \sum_{j=1}^{m-1} (b_j-b_m) \mu_j.
\]
Thus the variation of magnetic field through $j$-th hole contributes to the value of the spectral flow with coefficient $c_m(b_j-b_m)$.

If the signs of $B$ are the same on all boundary components, then the spectral flow is zero, 
no matter how magnetic field is varied 
(if only 
conditions (1-2) above are fullfiled). 
In the contrary, if $B$ takes positive values on some boundary component and negative values on another, then 
we can vary magnetic field so that the spectral flow does not vanish.
\end{rem}

\section{$2N$-dimensional Dirac operators}\label{sec:2NDO}

Let $X$ be as in the previous section.
The standard $2N$-dimensional Dirac operator has the form
\begin{equation}\label{eq:DO}
\Dir=-i \br{ \sigma_1 \px{} + \sigma_2 \py{} },
\; \mbox{ where }
\sigma_1 =
\matr{
    0 & I_N \\
    I_N & 0
}, \;
\sigma_2 =
\matr{
    0 & -iI_N \\
    iI_N & 0
},
\end{equation}
$I_N$ is $N\times N$ unit matrix.

We will consider operators of the form $\Dt = \Dir+Q(\x)$ acting on spinor functions 
\begin{equation}\label{eq:u}
u=\matr{ u^+ \\ u^- }, \; u^{\pm}\colon X\to\CC^N, 
\end{equation}
where $Q$ is a smooth map from $X$ to the space $\Herm(\CC^{2N})$ of complex self-adjoint (or, what is the same, Hermitian) $2N\times 2N$ matrices.

A self-adjoint local elliptic boundary condition for the operator $\Dir+Q$ has the form
\begin{equation}\label{eq:2Nbc}
    in(x)u^+ =  B(\x) u^- \; \mbox{ on } \p X,
\end{equation}
where $B$ is a smooth map from $\p X$ to the space of complex self-adjoint invertible $N\times N$ matrices,
$\n=\br{\n_1,\n_2}$ is the outward conormal to $\p X$ at point $\x$, 
and we identify $n$ with the complex number $\n_1+i\n_2$.

The equivalent way of imposing boundary condition \eqref{eq:2Nbc} is
\begin{equation}\label{eq:2N2bc}
    \br{ i\br{\n_1\sigma_1+\n_2\sigma_2} + \matr{B\inv & 0 \\ 0 & -B}} u = 0 \; \mbox{ on } \p X.
\end{equation}

\begin{thm}\label{thm:2NDirac}
Let $Q_t(\x)$ be a continuous 1-parameter family of self-adjoint $2N\times 2N$ matrices smoothly dependent on $\x\in X$ such that
$\Dir+Q_1 = \mu  \br{\Dir+Q_0}\mu \inv$ for some smooth gauge transformation $\mu \colon X \to U(1)$.
Let $B$ be a smooth map from $\p X$ to the space of complex self-adjoint invertible $N\times N$ matrices.
Then the spectral flow of the family $(\Dir+Q_t)$ with boundary condition \eqref{eq:2Nbc} is equal to
\[c_m \sumj b_j \mu_j,\]
where $c_m$ is the integer constant depending on $m$ only,
$\mu_j$ is the degree of the restriction of $\mu$ to the $j$-th connected component $\p X_j$ of the boundary,
$b_j$ is the number of negative eigenvalues of $B$ (counting multiplicities) on $\p X_j$ (this number is correctly defined due to nondegeneracy of $B$).
\end{thm}

This result is the corollary of Theorem \ref{thm:gDirac} from the following section.

\section{Dirac type operators}\label{sec:gDO}

Let $X$ be a compact planar domain bounded by $m$ smooth curves
and equipped with a Riemannian metric $g$ (which is not necessarily flat).

We call a first order formally self-adjoint operator $\Dt$ over $X$ a \textit{Dirac type operator} if its symbol has the form 
\begin{equation}\label{eq:sym_t}
\rho = \matr{ \rho_1 \\ \rho_2 } =
\Phi(\x) \matr{ \sigma_1 \\ \sigma_2 },
\end{equation}
where $\Phi$ is a smooth map from $X$ to the group $\GL^+(2,\R)$ of real invertible $2\times 2$ matrices with
positive determinant, and the matrices $\sigma_1$, $\sigma_2$ are defined by formula \eqref{eq:DO}.

In other words, Dirac type operator is the operator acting on spinor functions \eqref{eq:u} and having the following form:
\begin{equation}\label{eq:gDO}
\Dt = \Dt_{\Phi,\: Q} = -i \br{ \rho_1(x) \p_1 + \rho_2(x) \p_2 } + iR_{\Phi}(\x) + Q(\x),
\end{equation}
where $Q$ is a smooth map from $X$ to the space $\Herm(\CC^{2N})$ of complex self-adjoint $2N\times 2N$ matrices, 
$$R_{\Phi} = \frac{1}{2}\brr{ \br{ \rho_1\p_1 + \rho_2\p_2 } + \br{ \rho_1\p_1 + \rho_2\p_2 }^t }$$
(superscript $t$ denotes the operation of taking the formal adjoint operator).
More explicitely,
$$R_{\Phi}(\x) = -\frac{1}{2}\brr{ \px{\br{\sqg\rho_1}}+\py{\br{\sqg\rho_2}} } \in \Herm(\CC^{2N}),$$
where $\sqg =\sqrt{det(g_{ij})}$,
the matrix $(g_{ij})$ is inverse to the matrix $(g^{ij}) = (\bra{\d x^i, \d x^j}_g)$,
$\sqg\;\d x^1\d x^2$ is the volume element on $X$
(of course, $g_{ij}$, $g^{ij}$, and $\sqg$ depend on $x$).
\\

By $\DD=\DD\XgN$ we denote the space of all operators having the form \eqref{eq:gDO} for fixed $X$, $g$, $N$.
Note that Dirac type operator (that is an element of $\DD\XgN$)
is uniquely defined by the pair $(\Phi,Q)$.
\\

A self-adjoint elliptic local boundary condition for Dirac type operator \eqref{eq:gDO} has the form
\begin{equation}\label{eq:gbc}
    in'(x)u^+ =  B(\x) u^- \; \mbox{ on } \p X,
\end{equation}
where $B$ is a smooth map from $\p X$ to the space of complex self-adjoint invertible $N\times N$ matrices,
the complex-valued function $n'$ on $\p X$ is defined by the formula
$n' = n'_1 + in'_2$ with $(n'_1, n'_2) = (n_1, n_2)\Phi$ and
$n=(n_1,n_2)$ being the outward conormal to $\p X$ at $\x\in \p X$.
Recall that $\n_i=\sum g_{ij}\n^j$ for the components $\br{\n^j}$ of the normal to the boundary.

\begin{rem}
Equation \eqref{eq:gbc} is just another form of the equation
\begin{equation}\label{eq:gbc2}
i\rho^+(x,\n(x))u^+ =  B(\x) u^- \; \mbox{ on } \p X,	
\end{equation}
where 
\[
\rho(x,\xi) = 
\matr{0 & \rho^-(x,\xi) \\ \rho^+(x,\xi) & 0}
\]
denotes the symbol $\xi_1 \rho_1(x) + \xi_2 \rho_2(x)$ of the operator $D$ in the direction of a covector $\xi=(\xi_1,\xi_2)$.
Considering that in our case the operator $\rho^+(x,\xi)$ is scalar, and $\rho^+(x,n(x)) = n'(x)I_N$, 
boundary condition \eqref{eq:gbc2} may be written in simplified form \eqref{eq:gbc}.
\end{rem}

By $\B=\B\XN$ we denote the space of all smooth maps from $\p X$ to the space of complex self-adjoint invertible $N\times N$ matrices.
\\

Suppose $\Dt \in \DD$, $B\in \B$. We will write $(\Dt, B)$ for operator \eqref{eq:gDO} acting on the domain
$$\set{u\in C^1(X, \CC^{2N}) \colon \mbox{ restriction of } u \mbox{ to } \p X \mbox{ satisfies boundary condition \eqref{eq:gbc}} },$$
where $C^1(X, \CC^{2N})$ is the space of continuously differentiable functions from $X$ to $\CC^{2N}$.
\\

Such operators have the following properties:
\begin{enumerate}
    \item For any $\Dt \in \DD$, $B\in \B$ the operator $(\Dt, B)$ is (unbounded) essentially self-adjoint Fredholm operator,
    which has the discrete real spectrum. All its eigenvectors are smooth functions.
    (Lemma \ref{lem:1}, Section \ref{sec:prop})
    \item Suppose $Q_t(\x)$ is continuous on $(t, \x)$, $\Dt \in \DD$, $B\in \B$.
    Then all the operators from the family $\br{\Dt+Q_t,B}$ have the same domain,
    and this family is norm continuous in $L^2\br{X,g; \CC^{2N}}$.
    Therefore the spectral flow of the operators family $\br{\Dt+Q_t,B}$ is well defined (\cite{L-04}, Proposition 2.2).
\end{enumerate}

Now we can formulate the main result of the present paper:

\begin{thm}\label{thm:gDirac}
Let
$\Dt\in\DD$ be a Dirac type operator \eqref{eq:gDO},
$B\in\B$ define boundary condition \eqref{eq:gbc} for $\Dt$,
$Q_t(\x)$ be a continuous 1-parameter family of self-adjoint $2N\times 2N$ matrices smoothly dependent on $\x\in X$ such that
$\Dt+Q_1 = \mu  \br{\Dt+Q_0}\mu \inv$ for some smooth gauge transformation $\mu \colon X \to U(1)$.
Then
\[
\sp{\Dt+Q_t, B}\tin = c_m \sumj b_j \mu_j,
\]
where $c_m$ is the integer constant depending on $m$ only,
$b_j$ is the number of negative eigenvalues of $B$ (counting multiplicities) on $\p X_j$,
$\mu_j$ is the degree of the restriction of $\mu $ to $\p X_j$,
$\p X$ is oriented as described in the statement of Theorem \ref{thm:Dirac}.
\end{thm}

Note that constant $c_m$ in all the Theorems \ref{thm:Dirac}-\ref{thm:gDirac} is the same.


\begin{rem}
Let $S$ be a spinor bundle over $X$,
$\bra{\,\cdot\, ,\,\cdot\,}$ be an Hermitean metric on $S$ compatible with its spinor structure,
$\nabla$ be a connection on $S$ compatible with its spinor structure and the Levi-Civita connection on $TX$.
The Dirac operator on $S$ in local coordinates has the form 
$\Dt = \cl(v)\nabla_{v} + \cl(w)\nabla_{w}$, 
where $(v,w)$ is a positive oriented orthonormal basis in $T_x X$, 
and by $\cl(v)$ we denote the action of a tangent vector $v$ on spinors.

The unitary skew-adjoint isomorphism $J_x=\cl(v)\cl(w)$ of $S_x$ does not depend on the choice of a basis $(v,w)$ in $T_x X$ and defines the bundle decomposition $S = S^+\oplus S^-$, where $S^{\pm}$ are the subbundles of $S$ such that $S^{\pm}_x$ are the eigenspaces of $J_x$ corresponding to its eigenvalues $\mp i$. 
Due to the triviality of $TX$ and of any complex bundle over $X$, 
we can fix some global positive oriented orthonormal basis field $(v(x),w(x))$ in $TX$ and some trivialization of $S^-$.
Let us extend the trivialization from $S^-$ to $S$ so that the action of the tangent vectors on the spinors in this trivialization has the form
$$
\cl(v(x)) =
-i\matr{
    0 & I_N \\
    I_N & 0
}, \quad
\cl(w(x)) = -i\matr{
    0 & -iI_N \\
    iI_N & 0
}.
$$
Then sections $u$ of the spinor bundle $S$ can be identified with the column vectors \eqref{eq:u} 
of two functions $u^{\pm}\colon X\to\CC^N$, 
and Dirac operator $\Dt$ acting on such column vectors
can be written in the form
$\Dt = -i \br{\rho_1\nabla_1+\rho_2\nabla_2}$,
where $\rho_1$, $\rho_2$ are defined by formula \eqref{eq:sym_t}, 
$\Phi(\x)$ is the transition matrix:
$(v,w) = (e_1, e_2)\Phi(\x)$, and by $e_i$ we denote the vector (not the differential operator) $\p_i$ to avoid misunderstanding.

So any Dirac operator over $X$ has the form \eqref{eq:gDO} with $\Phi(\x)$ satisfying the condition 
$\Phi(\x)\Phi^{\ast}(\x)= (g^{ij}(\x))$ and with the matrix $Q(\x)$ of very special kind. 
While considering arbitrary metric $g$ and arbitrary $Q(\x)$ is important for some physical applications, 
considering Clifford multiplication which does not agree with the metric on $X$
(that is matrix function $\Phi(\x)$ which does not satisfy the condition $\Phi(\x)\Phi^{\ast}(\x)= (g^{ij}(\x))$)
does not seem neseccary.
Nevertheless, we take care of this more general case because the proofs of our results crucially depend on its consideration.
\end{rem}

\section{The case of one hole}\label{sec:annulus}

Here we compute the spectral flow for the case when $X$ has just one hole ($m=2$), 
and as a result find $c_2$.

\begin{thm}\label{thm:c2}
$c_2=1$.
\end{thm}

\begin{proof}
By Theorem \ref{thm:gDirac}, the spectral flow does not depend on the geometry of $X$ and on the choice of $\Dt\in\DD$,
so it is sufficient to consider only the case when the computation is as simple as possible.
Let us take the annulus $X=\set{\br{r,\varphi}\colon 1\leq r \leq 2}$ in the polar coordinates $\br{r,\varphi}$ on the plane, with the metric $\d s^2 = \d r^2 + \d\varphi^2$, $N=1$,
\[
\Dt = -i \matr{0 & \p_r-i\p_{\varphi} \\ \p_r+i\p_{\varphi} & 0}, \quad
\mu = e^{i\varphi}, \quad
Q_t = \matr{0 & it \\ -it & 0},
B = \sys{ {ll} +1 \mbox{ at } r=1 \\ -1 \mbox{ at } r=2}
\]
We obtain the following system for the eigenvector $u$ and the eigenvalue $\lambda$ of $\br{\Dt+Q_t,B}$:
\[
\sys{l
(-i\p_r + \p_{\varphi} - it) u^+ = \lambda u^-  \\
(-i\p_r - \p_{\varphi} + it) u^- = \lambda u^+  \\
u^+ = iu^- \mbox{ at } r=1,2
}
\]
All the eigenvectors of $(\Dt+Q_t, B)$ are smooth functions, so we can seek them in the form
$u^{\pm}(r,\varphi) = \sum_{k\in \Z}u^{\pm}_k(r)e^{ik\varphi}$.
Substituting it in the last system, we obtain
\[
\sys{l
\p_r u^+_k - (k-t)u^+_k - i\lambda u^-_k = 0 \\
\p_r u^-_k + (k-t)u^-_k - i\lambda u^+_k = 0 \\
u^+_k = iu^-_k \mbox{ at } r=1,2
}
\]
Equivalently,
\[
\sys{l
\p_r \br{u^+_k +i u^-_k} = (k-t-\lambda)\br{u^+_k -i u^-_k} \\
\p_r \br{u^+_k -i u^-_k} = (k-t+\lambda)\br{u^+_k +i u^-_k} \\
u^+_k - iu^-_k = 0 \mbox{ at } r=1,2
}
\]
and $\p_r^2 \br{u^+_k - i u^-_k} = \br{(k-t)^2-\lambda^2}\br{u^+_k - i u^-_k}$.
So we have the following cases:
\begin{itemize}
    \item either $u^+_k = u^-_k \equiv 0$,
    \item or $k-t+\lambda=0$, $u^-_k = \const$, $u^+_k = i u^-_k$,
    \item or $(k-t)^2-\lambda^2 = -(\pi l)^2$, $l\in\Z\setminus\set{0}$, $u^+_k - iu^-_k = \const\cdot\br{ e^{\pi ilr} - e^{-\pi ilr)} }$.
\end{itemize}
Therefore the set
$$\Lambda = \set{ (t,\lambda) \colon \lambda \mbox{ is the eigenvalue of } (\Dt+Q_t,B)}$$
is the union of the set $\Lambda_1=\set{ (t,\lambda) \colon \lambda-t\in\Z }$ (with the multiplicities 1 of the eigenvalues)
and of the set $\Lambda_2$ lying beyond the band $\left| \lambda \right| \geq\pi$.

If $\lambda_j(t)$ are the continuous functions from the interval $[0,1]$ to $\R$ such that
$\Lambda\cap\set{0\leq t \leq 1}$ is the union of the graphs of functions $\lambda_j(t)$
and $\lambda_i(t) \leq \lambda_j(t)$ for $i \leq j$,
then $\lambda_j(t) = j+t$ when $-3\leq j \leq 2$ (up to shift of the numeration).
So
$$\sp{\Dt+Q_t,B}\tin = 1.$$

On the other hand, by Theorem \ref{thm:gDirac},
$$\sp{\Dt+Q_t,B}\tin = c_2 \br{b_1 \mu_1 + b_2 \mu_2} = c_2\mu_2 = c_2$$
where by $\partial X_1$, $\partial X_2$ we denote the inner and the outer boundary circles respectively. 
Therefore $c_2=1$.

\end{proof}

\section{The case of several holes}\label{sec:manyholes}

In this section we provide some evidence supporting the conjecture that $c_m=1$ for all $m$.

Namely, let us realize $X=X^h$ as $(m-1)$ identical annuli arranged along the line and connected by the band of the width $h$,
with the corners smoothed out, as on the Fig. \ref{fig:multi1}.

\begin{figure}[!h]
\begin{center}
\includegraphics[width=0.6\textwidth]{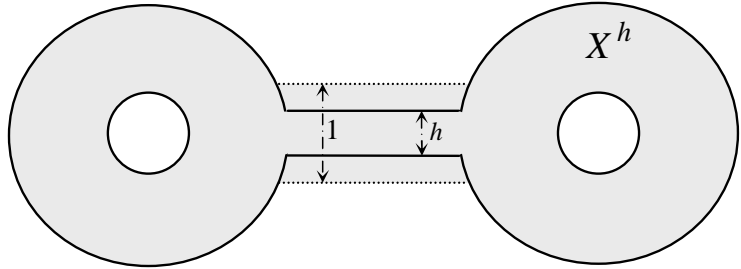}
\caption{Contracting of the connecting band}
\label{fig:multi1}
\end{center}
\end{figure}

Let us consider the process of continuous decreasing of the band's
width from $h=1$ to $h=0$; we suppose that the annuli do not change
in progress. Let us fix some function $\mu$ from $X^1$ to $U(1)$
and take $q_t = it\mu\inv\br{ \px{\mu }+ i\py{\mu } }$.
Restricting $\mu$ and $q_t$ on $X^h$, $0<h\leq 1$, we obtain
operator \eqref{eq:Qt} over $X^h$. Let us define the boundary
condition by $B^h=+1$ at the inner part $\cup_{j<m}\p X^h_j$ of
$\p X^h$ and $B^h=-1$ at the outer part $\p X^h_m$ of $\p X^h$.

By Theorem \ref{thm:Dirac}, $\sp{\Dir+Q_t} = c_m \sum_{j<m} \mu_j$ does not depend on $h$.
It is natural to suggest that the limit at $h\to +0$ of the (constant) spectral flow of the family $\br{\Dir+Q_t,B^h}$ for $X^h$
is equal to the spectral flow of $\br{\Dir+Q_t,B^0}$ for the \q{limit} domain $X^0$, which is the disjoint union of $m-1$ annuli,
and the ``limit'' boundary condition $B^0=+1$ at the inner boundary and $B^0=-1$ at the outer boundary of every annulus.

However, $\sp{\Dir+Q_t,B^0}$ for such union is the sum of $\sp{\Dir+Q_t,B^0}$ for the annuli,
and hence is equal to $\sum_{j<m} c_2\mu_j = c_2 \sum_{j<m} \mu_j$.

Therefore, if the assumption on the limit behavior of the spectral flow is true, then $c_m=c_2$ at any $m>2$.

\begin{figure}[!h]
\begin{center}
\includegraphics[width=0.8\textwidth]{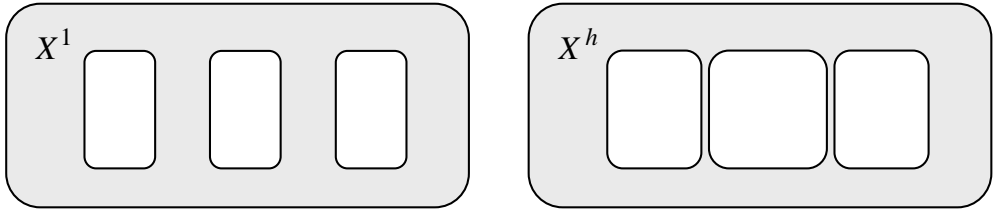}
\caption{Increasing of the holes}
\label{fig:multi2}
\end{center}
\end{figure}

Another way to have a look at the general case is to fix the outer boundary and to increase the holes up to their merging,
as on Fig. \ref{fig:multi2}.
Here we obtain the single annulus in the limit of $h=0$,
and the same result $c_m=1$ if the passage to the limit will be justified.

Alternatively, we can combine these two methods to obtain
arbitrary number $m'$, $1\leq m'\leq m-1$ of annuli in the end of
the limit process, with the same result for $c_m$.

\section{General case: first order elliptic operators}\label{elliptic}

The results of the present paper are concerned only with the case when $X$ is a disk with holes.
Easy modification of the proof gives us analogue of this result for the case of smooth compact oriented surface $X$ with nonempty boundary,
with the only change of $c_m$ to $c_{m,\,g}$, where $c_{m,\,g}$ is the integer constant
depending on the number $m$ of boundary components of $X$ and on the genus $g$ of $X$.
However, this still remains within the very restricted framework:
all the operators $\Dt_t$ are of Dirac type, both symbol of $\Dt_t$ and boundary condition do not depend on $t$,
conjugating gauge transformation is scalar.

In fact, this result can be extended to much more general case.
Namely, let $X$ be a smooth compact surface,
$(A_t)$ be an 1-parameter family of first order symmetric elliptic differential operators
acting on sections of unitary vector bundle $E$ over $X$,
and subbundle $L_t$ of $\EY$ defines a self-adjoint elliptic local boundary condition for $A_t$ at any $t$.
Suppose that $(A_1,L_1)$ is conjugate to $(A_0,L_0)$ by some gauge transformation $\mu$
(that is $\mu$ is unitary isomorphism of $E$, not necessarily scalar).
Then operators $(A_1,L_1)$, $(A_0,L_0)$ are isospectral,
and there arises the natural question about the spectral flow of the family $(A_t,L_t)$.
This question will be considered in a forthcoming paper by the author \cite{Prokh}. 
In that paper we will prove that
$$\sp{A_t, L_t}\tin = c_{m,\,g} \sumj \varphi_j,$$
where $c_{m,\,g}$ is the integer depending on the number $m$ of boundary components of $X$ and on the genus $g$ of $X$,
$\varphi_j$ is the integer determined in a canonical way by the restrictions on $j$-th boundary component of the following data:
\vspace{-5pt}
\begin{enumerate}
\parsep=0pt 
\itemsep=\parsep
    \item[(1)] family $(\rho_t)$, where $\rho_t$ is the symbol of $A_t$;
    \item[(2)] family $(L_t)$ of boundary conditions;
    \item[(3)] gauge transformation $\mu$.
\end{enumerate}
In particular, the spectral flow of $(A_t,L_t)$ does not depend on the choice of the operators in the interior of $X$ but only on the symbol of the operators on the boundary.

Theorem \ref{thm:gDirac} of present paper fits into this general result as follows: $c_m=c_{m,\,0}$, $\varphi_j = b_j \mu_j$.
Recall that $\mu_j$ is invariant of the restriction of $\mu$ to $j$-th boundary component of $X$,
$b_j$ is defined from the restrictions of boundary condition and of the operator's symbol on $j$-th boundary component.

\section{The spectral flow for $N=2$ in terms of the condensed matter physics}\label{sec:phys}

In this section we compare our boundary condition \eqref{eq:2Nbc} with 
\q{the general boundary conditions for the Dirac equation} for 4-dimensional Dirac operators formulated by Akhmerov and Beenakker in \cite{AhmBeen}. 
After that, we give some computations for the spectral flow in terms of \cite{AhmBeen}.
In particular, we show that the spectral flow vanishes in the case of time reversal symmetry 
(under the assumption of local ellipticity of the boundary problem).

In this section we will temporarily use the notations from \cite{AhmBeen} in their original form 
and will formulate our results in the same terms. 

The long-wavelength and low-energy electronic excitations in graphene 
considered in \cite{AhmBeen} 
are described by the Dirac equation $H\Psi=\varepsilon\Psi$ with Hamiltonian
\begin{equation}\label{eq:DirHam}
H=v\tau_0\otimes(\bm{\sigma}\cdot \bm{p}) 
\end{equation}
acting on a four-component spinor wave function $\Psi=(\Psi_A,\Psi_B)$
(in our notations, $\Psi$ is two-dimensional spinor function, $\Psi_A=u^+$, $\Psi_B=u^-$, $N=2$).
Here $v$ is the Fermi velocity, $\bm{p}=-i\hbar\nabla$ is the momentum operator,
$\bm{\sigma}\cdot \bm{p} = -i\hbar\br{\sigma_1\nabla_1 + \sigma_2 \nabla_2}$, 
matrices $\tau_i, \sigma_i$ are Pauli matrices in valley space and sublattice space, respectively:
\[
\sigma_0 = \matr{1 & 0 \\ 0 & 1}, \;
\sigma_1 = \matr{0 & 1 \\ 1 & 0}, \;
\sigma_2 = \matr{0 & -i \\ i & 0}, \;
\sigma_3 = \matr{1 & 0 \\ 0 & -1}, \;
\tau_i=\sigma_i.
\]
The general energy-independent boundary condition posed in \cite{AhmBeen} has the form
\begin{equation}\label{eq:MPsi}
\Psi = M \Psi \mbox{ on the boundary,}  
\end{equation}
where $M$ is a self-adjoint unitary $4\times 4$ matrix depending on the point $x\in \p X$ 
and anticommuting with the current operator 
$v\tau_0\otimes(\bm{\sigma} \cdot \bm{n}_B)$.
Here $\bm{n}_B$ is the outward normal to the boundary of $X$ at $x$ , 
so $\bm{n}_B = (n_1,n_2)$ in our previous notations, and
$\bm{\sigma}\cdot \bm{n}_B = n_1\sigma_1 + n_2\sigma_2$.

Let us compare \eqref{eq:MPsi} with our boundary condition \eqref{eq:2Nbc}.

At first note that the condition \q{$M$ is a self-adjoint unitary matrix anticommuting with the current operator} 
mean nothing but the condition of self-adjointness of the boundary problem \eqref{eq:MPsi}.
The authors of \cite{AhmBeen} do not require local ellipticity from the boundary condition; 
however, in the absence of local ellipticity the spectrum of the operator is not expected to be discrete. 
Boundary condition \eqref{eq:MPsi} is both locally elliptic and self-adjoint 
if and only if the matrix function $M(x)$ can be represented by the formula
\begin{equation*}
    M = I_{2N} - 2\matr{ I_N+B^2 & 0 \\ 0 & I_N+B^2 }\inv \matr{ I_N & i\bar{n}B \\ -in B & B^2 }
\end{equation*}
for some complex self-adjoint invertible $N\times N$ matrix function $B(\x)$. 
For such $M$ boundary condition \eqref{eq:MPsi} is equivalent to our boundary condition \eqref{eq:2N2bc}. 

The set of all possible self-adjoint unitary $4\times 4$ matrices anticommuting with the current operator 
is paramertrised in \cite{AhmBeen} by the following 4-parameter family: 
\begin{equation}\label{eq:M}
M = \sin\Lambda\;\tau_0 \otimes (\bm{n}_1\cdot\bm{\sigma}) + \cos\Lambda\;(\bm{\nu} \cdot \bm{\tau})\otimes (\bm{n}_2\cdot\bm{\sigma}),
\end{equation}
where $\Lambda(x)$ is a \q{mixing angle}, 
$\bm{\nu}(x)$, $\bm{n}_1(x)$, $\bm{n}_2(x)$ are unit vectors in $\R^3=\set{(x^1,x^2,x^3)}$ 
such that $\bm{n}_1$ and $\bm{n}_2$ are mutually orthogonal and also orthogonal to the boundary normal $\bm{n}_B(x)$,
$(\bm{\nu} \cdot \bm{\tau}) = \sum_{i=1}^3 {\nu^i \tau_i}$, and $(\bm{n}_j\cdot\bm{\sigma})$ are defined analogously.
Now we give the description of the ellipticity of the boundary problem \eqref{eq:MPsi}
in terms of $\Lambda$, $\bm{\nu}$, $\bm{n}_1$, $\bm{n}_2$, 
and compute $b_j$ as the functions of these parameters.

From now on we will suppose that the frame $(\bm{n}_B, \bm{n}_1, \bm{n}_2)$ is \textsl{positive oriented} in $\R^3$, that is its orientation coincides with the orientation of the frame $(e_1,e_2,e_3)$ of basis coordinate vectors.
This is possible because parameters 
$(\Lambda, \bm{n}_1, \bm{n}_2, \bm{\nu})$ and $(-\Lambda, -\bm{n}_1, \bm{n}_2, \bm{\nu})$ 
give us the same matrix $M$, so in a case of a negative oriented frame $(\bm{n}_B, \bm{n}_1, \bm{n}_2)$
we can change the signs of $\bm{n}_1$ and $\Lambda$ simultaneously to obtain the positive orientation of the frame.

Let $\varphi(x)$ be a function from the boundary to the circle $\R\mod 2\pi$ such that
$\bm{n}_2 = \sin\varphi\cdot\eta+\cos\varphi\cdot e_3$, 
where $e_3$ is the unit vector in $\R^3$ in the direction of $x^3$,
$\eta(x)$ is the unit tangent vector to the boundary at $x\in\p X$ 
such that the pair $(\bm{n}_B(x),\eta(x)$) has the positive orientation on the plane $(x^1,x^2)$.
Then $\bm{n}_1 = \cos\varphi\cdot\eta-\sin\varphi\cdot e_3$,  
and $M$ is determined by the triple $(\Lambda, \varphi, \bm{\nu})$.

\begin{prop}\label{prop:1}
Boundary condition \eqref{eq:MPsi} is locally elliptic for Dirac operator \eqref{eq:DirHam} if and only if 
$\Lambda+\varphi \neq 0 \,(\mathrm{mod}\,\pi)$ and $\Lambda-\varphi \neq 0 \,(\mathrm{mod}\,\pi)$ for any $x\in\p X$.
\end{prop}

In other words, the boundary condition is not locally elliptic if and only if 
$\bm{n}_2 = \pm\sin\Lambda\,\eta \pm\cos\Lambda\, e_3$ for some $x\in\p X$ and for some combination of signs $\pm$.

\begin{prop}\label{prop:2}
If boundary condition \eqref{eq:MPsi} is locally elliptic for Dirac operator \eqref{eq:DirHam}
then it is equivalent to the boundary condition 
\[
in(x)\Psi_A = B(x)\Psi_B
\]
with the matrix $B(x)$ defined as following:
\[
B = \beta_+ P_+ + \beta_- P_-, \; \mbox{ where }
 P_{\pm} = \frac{\tau_0 \pm (\bm{\nu} \cdot \bm{\tau})}{2}, \;
\beta_+ = \ctg\frac{\Lambda+\varphi}{2}, \;
\beta_- = \tg\frac{\Lambda-\varphi}{2}.
\]
Here $\beta_{\pm}$ are the eigenvalues of $B$, 
$P_{\pm}$ are the ortogonal projections on the invariant subspaces of $B$ corresponding to the eigenvalues $\beta_{\pm}$.
\end{prop}

We prove these Propositions in the end of the section.

\begin{cor}\label{cor:1}
Let $Q_t(\x)$ be a continuous 1-parameter family of self-adjoint $4\times 4$ matrices smoothly dependent on $\x\in X$ such that $H+Q_1 = \mu  \br{H+Q_0}\mu \inv$ for some smooth gauge transformation $\mu \colon X \to U(1)$.
Suppose that boundary condition \eqref{eq:MPsi} is locally elliptic for Dirac operator \eqref{eq:DirHam}.
Then the spectral flow of the family $(H+Q_t)$ with this boundary condition is described by the formulae
\[
\sp{H+Q_t,M}\tin=c_m \sumj b_j \mu_j,
\]
where $c_m$, $\mu_j$ are as in Theorem \ref{thm:2NDirac}, 
$b_j$ depends only on the values of $\Lambda$, $\varphi$ on $j$-th boundary component:
\[
b_j=\case{
  0, \mbox{ if both }\Lambda+\varphi, \; \Lambda-\varphi \mbox{ belong to the interval } (0,\pi) \\
  2, \mbox{ if both }\Lambda+\varphi, \; \Lambda-\varphi \mbox{ belong to the interval } (\pi,2\pi) \\
  1, \mbox{ if one of }\Lambda+\varphi, \; \Lambda-\varphi \mbox{ belongs to the interval } (0,\pi) 
      \mbox{ and another to the interval } (\pi,2\pi)
 } 
\]
\end{cor}

\proof
This follows immediately from Theorem \ref{thm:2NDirac} and Proposition \ref{prop:2}.
\\

Let us inspect closer the case of time reversal symmetry.
The time reversal operator in the valley isotropic representation is
\begin{equation*}
T=-(\tau_2\otimes\sigma_2){\cal C},
\end{equation*}
with ${\cal C}$ the operator of complex conjugation \cite{AhmBeen}. 
The boundary condition preserves time reversal symmetry if $M$ commutes with $T$. 
This implies that the mixing angle $\Lambda\equiv 0$ \cite{AhmBeen}.
By Proposition \ref{prop:1}, in this case boundary problem \eqref{eq:MPsi} is locally elliptic if and only if 
$\bm{n}_2(x)$ is not vertical for all $x\in\p X$.
If this is fullfilled then 
Corollary \ref{cor:1} allows us to compute the spectral flow regardless of other parameters:

\begin{cor}[the case of time reversal symmetry]
\label{cor:2}
Let $Q_t(\x)$, $t\in [0,1]$ be a continuous 1-parameter family of self-adjoint $2N\times 2N$ matrices smoothly dependent on $\x\in X$ such that $H+Q_1 = \mu  \br{H+Q_0}\mu \inv$ for some smooth gauge transformation $\mu \colon X \to U(1)$.
Suppose that boundary condition is defined by formulas \eqref{eq:MPsi}, \eqref{eq:M} with $\Lambda\equiv 0$, 
and that for any $x\in\p X$ vector $\bm{n}_2(x)$ is not vertical.
Then the spectral flow of the family $(H+Q_t)$ is zero.
\end{cor}

\proof
By Corollary \ref{cor:1}, $b_j=1$ for all $j$.
So we obtain 
$$\sp{H+Q_t}\tin = c_m \sumj b_j \mu_j = c_m \sumj \mu_j = 0.$$

\begin{rem}
Corollary \ref{cor:2} can be proved by other means as well, without use of Corollary \ref{cor:1} and of formula \eqref{eq:M} but using Theorem \ref{thm:2NDirac} directly.
Namely, let $M\Psi=\Psi$ be locally elliptic boundary condition for Dirac operator \eqref{eq:DirHam} such that $TMT\inv=M$.
At first note that the spectral flow of the family $(H+Q_t)$ is independent of the choice of connection $\nabla$, 
so we can assume that $\nabla_i = \p_i$. With this choice of connection, we have $THT\inv=H$. 
Let $Q'_t(x)=TQ_t(x)T\inv$, then  
\[
Q'_1-Q'_0 = T(Q_1-Q_0)T\inv = T(\mu H\mu\inv-H)T\inv = \mu\inv THT\inv \mu - THT\inv = \mu' H\mu'^{-1}-H,
\]
where $\mu'=\mu\inv$.
By Theorem \ref{thm:2NDirac}, 
\[
\sp{H+Q'_t,M} = c_m \sum b_j \mu'_j = -c_m \sum b_j \mu_j = -\sp{H+Q_t,M},
\]
where by $\sp{H+Q_t,M}$ we denote the spectral flow of the family $(H+Q_t)$ with boundary condition $M\Psi=\Psi$.
If $M$ commutes with $T$ then
\begin{equation}\label{eq:conj}
\sp{T(H+Q_t)T\inv, TMT\inv} = \sp{H+Q'_t,M} = -\sp{H+Q_t,M}.
\end{equation}
In Section \ref{sec:P0-P4} we prove the conjugacy invariance of the spectral flow 
under unitary isomorphisms of $L^2(X,g; \CC^{2N})$.
Even though $T$ is \textit{antilinear} isomorphism of the Hilbert space $L^2(X,g; \CC^4)$, 
the spectral flow still remains invariant under conjugation by $T$. 
This can be proved using the uniqueness property of the spectral flow in the same manner 
as Property (P4) in Section \ref{sec:P0-P4}, 
taking into account that conjugation by $T$ preserves self-adjointness of operators in $L^2(X,g; \CC^4)$.
Thus $\sp{T(H+Q_t)T\inv, TMT\inv}$ coincides with $\sp{H+Q_t,M}$. 
Together with \eqref{eq:conj}, this imply $\sp{H+Q_t,M} = 0$.
\end{rem}

\noindent\textit{Proof of Proposition \ref{prop:1}.}
By $S=S(x)$ we denote the matrix $(\bm{\nu} \cdot \bm{\tau}) = \sum_{i=1}^3 {\nu^i \tau_i}$.
In our notations,
\begin{multline*}
M = 
\sin\Lambda\cdot \tau_0 \otimes \matr{-\sin\varphi & -i\bar{n}\cos\varphi \\ in\cos\varphi & \sin\varphi} 
+ \cos\Lambda\cdot S \otimes \matr{\cos\varphi & -i\bar{n}\sin\varphi \\ in\sin\varphi & -\cos\varphi} = \\
= \matr{S_1 & -i\bar{n}S_2 \\ inS_2 & -S_1}, 
\end{multline*}
where 
$S_1 = -\sin\varphi\sin\Lambda\cdot I + \cos\varphi\cos\Lambda\cdot S$, 
$S_2 = \cos\varphi\sin\Lambda\cdot I + \sin\varphi\cos\Lambda\cdot S$,
$I=I_2$ is $2\times 2$ identity matrix.
Note that $S^2=I$ for any $\bm{\nu}$, $S_1^2+S_2^2=I$ for any $\bm{\nu}$, $\varphi$, $\Lambda$.

Boundary condition $M\Psi=\Psi$ is equivalent to the following system:
\begin{equation}\label{eq:bcS}
\sys{r
-i\bar{n}S_2 u^- = (I-S_1)u^+ \\
inS_2 u^+ = (I+S_1)u^-
}
\end{equation}
This boundary problem is locally elliptic for operator \eqref{eq:DirHam} if the linear space of the solutions of this system intersects with both spaces $\set{u^+=0}$ and $\set{u^-=0}$ by zero subspace.
This condition is equivalent to the invertibility of $S_2$. 
Matrix $S$ has the eigenvalues $\pm 1$, so $S_2$ has the eigenvalues 
$\cos\varphi\sin\Lambda \pm \sin\varphi\cos\Lambda = \sin(\Lambda\pm\varphi)$.
Both eigenvalues of $S_2$ are nonzero if and only if $\Lambda\pm\varphi \neq 0 \,(\mathrm{mod}\,\pi)$.
This completes the proof.

\medskip\noindent\textit{Proof of Proposition \ref{prop:2}.}
From \eqref{eq:bcS} we have $B = S_2\inv(I+S_1)$.
Taking into account the identity $S^2=I$, we obtain
\begin{multline*}
S_2\inv = \br{ \cos^2\varphi\sin^2\Lambda - \sin^2\varphi\cos^2\Lambda }\inv 
   \br{ \cos\varphi\sin\Lambda\cdot I - \sin\varphi\cos\Lambda\cdot S }, \\
S_2\inv(I+S_1) = 
 \br{ \cos^2\varphi\sin^2\Lambda - \sin^2\varphi\cos^2\Lambda }\inv \br{\sin\Lambda - \sin\varphi}
 \br{\cos\varphi\cdot I+\cos\Lambda\cdot S} = \\
 = \frac{\sin\Lambda-\sin\varphi}{\sin(\Lambda+\varphi)\sin(\Lambda-\varphi)}(\cos\varphi\cdot I + \cos\Lambda\cdot S).
\end{multline*}
Eigenvalues of $S$ are $\pm 1$,
so the eigenvalues of $B$ are equal to
\begin{equation}\label{eq:beta}
\beta_{\pm} = \frac{\sin\Lambda-\sin\varphi}{\sin(\Lambda+\varphi)\sin(\Lambda-\varphi)} (\cos\varphi\pm\cos\Lambda).
\end{equation}
From last two formulas we have 
\begin{multline*}
B = \frac{\sin\Lambda-\sin\varphi}{\sin(\Lambda+\varphi)\sin(\Lambda-\varphi)}
\br{ (\cos\varphi + \cos\Lambda)\frac{I+S}{2} + (\cos\varphi - \cos\Lambda)\frac{I-S}{2} } = \\
= \beta_+ \frac{I+S}{2} + \beta_- \frac{I-S}{2}
= \beta_+ P_+ + \beta_- P_-.
\end{multline*}
We can simplify \eqref{eq:beta} using sum-to-product trigonometric identities:
\[
\frac{\sin\Lambda-\sin\varphi}{\sin(\Lambda+\varphi)\sin(\Lambda-\varphi)} = 
\frac{ 2\sin\frac{\Lambda-\varphi}{2}\cos\frac{\Lambda+\varphi}{2}}
{\br{2\sin\frac{\Lambda+\varphi}{2}\cos\frac{\Lambda+\varphi}{2}} \br{2\sin\frac{\Lambda-\varphi}{2}\cos\frac{\Lambda-\varphi}{2}}} 
= \br{2\sin\frac{\Lambda+\varphi}{2}\cos\frac{\Lambda-\varphi}{2}}\inv, 
\]
\[
\cos\varphi + \cos\Lambda = 2\cos\frac{\Lambda+\varphi}{2}\cos\frac{\Lambda-\varphi}{2}, \quad
\cos\varphi - \cos\Lambda = 2\sin\frac{\Lambda+\varphi}{2}\sin\frac{\Lambda-\varphi}{2}.
\]
Substituting this in \eqref{eq:beta}, we obtain 
\[
\beta_+ = \ctg\frac{\Lambda+\varphi}{2}, \quad 
\beta_- = \tg\frac{\Lambda-\varphi}{2}.
\]
This completes the proof.

\part{Proof of Theorem \ref{thm:gDirac}}\label{part:2}

Note that for $\Dt'=\Dt+Q_0$, $Q'_t=Q_t-Q_0$ we have  
$\sp{\Dt+Q_t, B} = \sp{\Dt'+Q'_t, B}$ with $Q'_0=0$.
By this reason, in the proof we will restrict ourselves by the families $Q_t$ with $Q_0=0$. 

\section{Two technical lemmas}\label{sec:prop}

First of all, we need to give some technical details.
The reader interested only in the ideas behind the proof can go directly to the next section. 

Suppose $\Dt \in \DD$, $B\in \B$.
We will write $\DtB$ for the operator $\Dt$ acting on the domain
\begin{multline}\label{eq:domain}
  \dom\DtB = \\
  = \set{u\in L^2_1\br{X; \CC^{2N}}\colon \mbox{ restriction of } u \mbox{ to } \p X \mbox{ satisfies boundary condition \eqref{eq:gbc}} }.	
\end{multline}
Here $L^2_1\br{X; \CC^{2N}}$ is the first Sobolev space; its elements are functions $u\in L^2\br{X; \CC^{2N}}$ such that $\p_1 u, \p_2 u \in L^2\br{X; \CC^{2N}}$.
Strictly speaking, we use here not the restriction in the usual sense (trace map $u\mapsto \restr{u}{\p X}$)
but the extension by continuity of the trace map $C^{\infty}\br{X; \CC^{2N}} \to C^{\infty}\br{\p X; \CC^{2N}}$
to the bounded linear map from $L^2_1\br{X; \CC^{2N}}$ to $L^2_{1/2}\br{\p X; \CC^{2N}}$ \cite{BBLZ-09}.

Note that the operator $\DtB$ defined here is the closure of the operator $(\Dt, B)$ defined in Section \ref{sec:gDO}
(see \cite{BBLZ-09}, Proposition 2.9).
Using of non-closed operators in the first part of the paper is explained by our intention to avoid the introduction of Sobolev spaces and of the extension of the trace map as long as possible.
Due to the following Lemma, these two definitions give us the operators with the same eigenvectors, 
so this slight abuse of notation does not cause any troubles.

\begin{lem}\label{lem:1}
For any $\Dt \in \DD$, $B\in \B$ the operator $\DtB$ is
(unbounded) closed self-adjoint Fredholm operator on $L^2\br{X,g; \CC^{2N}}$,
which has the discrete real spectrum.
Moreover, all its eigenvectors are smooth functions.
\end{lem}

\proof
Let $B$ be a smooth function from $\p X$ to $\GL(N,\CC)$.
Then for any $\Dt\in\DD$, $\lambda\in\CC$ boundary condition \eqref{eq:gbc} satisfies the Sapiro-Lopatinskii condition for $\Dt-\lambda$: 
the intersections of the subspace $\set{u\colon in'(x)u^+ =  B(\x) u^-} \subset \CC^{2N}$ 
both with $\set{u\colon u^- = 0}$ and with $\set{u\colon u^+ = 0}$ are zero at any $x\in \p X$.
By Proposition 2.9 from \cite{BBLZ-09}, \eqref{eq:gbc} is strongly regular boundary condition for $\Dt$,
so all eigenvectors of $\DtB$ in $L^2\br{X,g; \CC^{2N}}$ are smooth functions.
By the same Proposition, $(\Dt-\lambda,B)$ is a closed Fredholm operator for any $\lambda\in\CC$,
so the spectrum of $\DtB$ is discrete.

For any $u,w \in L_1^2\br{X,g; \CC^{2N}}$ we have
\begin{multline*}
	\bra{\Dt u,w}_{L^2} - \bra{u, \Dt w}_{L^2} =
	\int_X \br{ \bra{\Dt u,w} - \bra{u, \Dt w} } \sqg\;\d x^1 \d x^2 = \\
	= -i \int_X {\br{ \p_1\bra{\sqg\rho_1 u,w} + \p_2\bra{\sqg\rho_2 u,w} } \d x^1 \d x^2} = \\
	= -i \int_X {\d \br{ \bra{\sqg\rho_1 u,w}\d x^2 - \bra{\sqg\rho_2 u,w}\d x^1 } } 
	= -i \int_{\p X} {\sqg \br{ \bra{\rho_1 u,w}\d x^2 - \bra{\rho_2 u,w}\d x^1 } } = \\
	= -i \int_{\p X} { \bra{ (n_1\rho_1+n_2\rho_2) u,w}\sqg\d s } 
	= -  \int_{\p X} { \bra{ \matr{i\bar{n}'u^- \\ in'u^+} , \matr{w^+\\w^-}}\sqg\d s }	= \\
	=    \int_{\p X} { \br{ \bra{ u^-, in'w^+} - \bra{in'u^+, w^-}}\sqg\d s },	
\end{multline*}
where $\d s$ is the length element on $\p X$.
So for any $u,w \in \dom\DtB$ 
$$\bra{\Dt u,w}_{L^2} - \bra{u, \Dt w}_{L^2} = \int_{\p X} { \bra{ u^-, (B-B^{\ast}) w^-}\sqg\d s },$$
and the operator $\DtB$ is symmetric on $L^2\br{X,g; \CC^{2N}}$ if and only if $B(\x)$ is self-adjoint at any $\x$.

Let now $w \in \dom\DtB^{\ast}$. 
By Proposition 2.9 from \cite{BBLZ-09}, $\dom\DtB^{\ast}$ is contained in $L_1^2\br{X,g; \CC^{2N}}$, so we can use the computation above:
$$\bra{\Dt u,w}_{L^2} - \bra{u, \Dt w}_{L^2} = \int_{\p X} { \bra{ u^-, (in'w^+ - Bw^-)}\sqg\d s }$$
for any $u \in \dom\DtB$.
Therefore, $\restr{in'w^+ - Bw^-}{\p X} = 0$ for any $w \in \dom\DtB^{\ast}$, $\dom\DtB^{\ast} = \dom\DtB$, 
and $\DtB$ is self-adjoint on $L^2\br{X,g; \CC^{2N}}$. 
All eigenvalues of a self-adjoint operator are real. This completes the proof. \\

In the statement of Theorem \ref{thm:gDirac} we used only norm continuous paths of operators with fixed domain. 
But for the proof of Theorem \ref{thm:gDirac} we have to deal with the paths of more general kind, 
when neither symbol of the operator nor boundary condition are fixed any more.
The paths we need for the proof are not norm continuous but only graph continuous 
(note that by Proposition 2.2 from \cite{L-04} any norm continuous path is graph continuous as well).
So further we will use the graph topology on the space of closed densely defined self-adjoint operators on 
a separable Hilbert space $H$ (in our case $H=L^2\br{X,g; \CC^{2N}}$).

There are various definitions of the graph distance, all of which give the same graph topology \cite{BLP-04}. 
One can take $d_G(A,A')=\norm{(A+iI)\inv-(A'+iI)\inv}$, 
or alternatively $d_G(A,A')=\norm{P_A-P_{A'}}$, where $P_A$, $P_{A'}$
are the orthogonal projections of $H\times H$ onto the graphs of $A$, $A'$ respectively.

Let us introduce the following metrics in $\DD$ and $\B$:
\begin{multline*}
d\br{ \Dt_{\Phi,\: Q}, \Dt_{\Phi',\: Q'} } = \norm{Q-Q'}_{C(X)} + \norm{\Phi-\Phi'}_{C^1(X)} = \\
= \max_{\x\in X} \norm{Q(\x)-Q'(\x)} + \max_{\x\in X} \br{ \norm{\Phi(\x)-\Phi'(\x)}+\norm{\p_1\Phi(\x)-\p_1\Phi'(\x)}+\norm{\p_2\Phi(\x)-\p_2\Phi'(\x)} },
\end{multline*}
\[
d(B,B') = \norm{B-B'}_{C^1(\p X)} =
 \max_{\x\in\p X} \br{ \norm{B(\x)-B'(\x)}+\norm{\p_s B(\x)-\p_s B'(\x)} },
\]
where $s$ is the length parameter on $\p X$.
Here we use any of the standard norms on the spaces $\End(\CC^{N})$ and $\End(\CC^{2N})$ of complex $N\times N$ and $2N\times 2N$ matrices, and on the space $\End(\R^2)$ of real $2\times 2$ matrices.

Note that $\br{\Dt_{\Phi_t,Q_t}, B_t}$ is the continuous path in $\DB$ if and only if
$Q_t(\x)$, $\Phi_t(\x)$, $B_t(\x)$, and the first partial derivatives of $\Phi_t(\x)$, $B_t(\x)$ with respect to $\x$
are continuous functions of $(t,\x)$.

Denote by $\HF(H)$ the space of closed self-adjoint (or, what is the same, Hermitian) Fredholm operators in separable Hilbert space $H$.
We fix graph topology on $\HF(H)$.
Nevertheless we will usually write \q{graph continuous} instead of mere \q{continuous} for the maps to $\HF(H)$ to avoid a misunderstanding.

By Lemma \ref{lem:1}, we have the natural inclusion $\DB \hookrightarrow \HF \br{L^2\br{X,g; \CC^{2N}}}$, which carries a pair $(D,B)\in\DB$ to the operator $D$ with the domain \eqref{eq:domain}. 

\begin{lem}\label{lem:2}
The natural inclusion $\DB \hookrightarrow \HF \br{L^2\br{X,g; \CC^{2N}}}$
is graph continuous.
\end{lem}

Therefore, if $t\mapsto\brop{\Dt_t,B_t}$ is the continuous path in $\DB$,
then $\brop{\Dt_t,B_t}$ defines the graph continuous path in $\HF \br{L^2\br{X,g; \CC^{2N}}}$,
and the spectral flow of the operators family $\brop{\Dt_t,B_t}$ is well defined.

\proof
Let us consider the smooth map
\[
\psi \colon \End(\CC^N) \to \Herm(\CC^{2N}), \quad
A \mapsto P =
\matr{ I_N & -A \\ -A^{\ast} & A^{\ast}A }
\matr{I_N+AA^{\ast} & 0 \\ 0 & I_N+A^{\ast}A }\inv,
\]
which carries $A\in\End(\CC^N)$ into the orthogonal projection $P$
of $\CC^{2N}$ with $\Ker P = \set{ u=\br{u^+,u^-} \colon u^+,u^- \in \CC^N, u^+ = Au^- }$.
It induces the continuous map
\[
\psi_{\ast} \colon C^1\br{\p X, \End(\CC^N)} \to C^1\br{\p X, \Herm(\CC^{2N})}.
\]
Composing $\psi_{\ast}$ with the continuous map 
$$\DB\to C^1\br{\p X, \End(\CC^N)}, \quad \DtB \mapsto 
-i\rho^+(x,n(x))\inv B(x),$$ 
we obtain the continuous map 
$$\Psi\colon\DB\to C^1\br{\p X, \Herm(\CC^{2N})},$$ 
which carries $(\Dt,B)$ into the orthogonal projection $P$ of $L^2\br{\p X, \restr{g}{\p X};
\CC^{2N}}$ with the kernel defined by boundary condition \eqref{eq:gbc} 
\footnote{Here we use the general formula
for the orthogonal projection $P$ with the kernel $\set{u^+ = Au^-}$ for arbitrary matrix $A$. 
Actually, in our case $A = (in')\inv B$ is normal: $AA^{\ast}=A^{\ast}A$.}.

By Proposition II.1.1 from \cite{Str-67}, we have the continuous inclusion of the Banach spaces
\[
\End\br{L^2_1\br{\p X; \CC^{2N}}} \subset \End\br{L^2_{1/2}\br{\p X; \CC^{2N}}},
\]
where $\End(V)$ denotes the space of bounded linear operators on a Banach space $V$,
$L^2_{\,r}$ is the (fractional) Sobolev space.
Composing it with the natural continuous inclusion 
\[
C^1\br{\p X, \End(\CC^{2N})} \subset \End\br{L^2_1\br{\p X; \CC^{2N}}},
\]
we obtain that the map $\Psi_{\ast} \colon\DB \to \End\br{L^2_{1/2}\br{\p X; \CC^{2N}}}$ is continuous.

The natural map from $\DD$ to the space of bounded linear operators
from $L^2_1\br{X; \CC^{2N}}$ to $L^2\br{X; \CC^{2N}}$ is continuous too:
\begin{multline*}
\norm{\Dt_{\Phi,Q}-\Dt_{\Phi',Q'}}_{1,0} \leq
\const\br{ \norm{\Phi-\Phi'}_{C(X)} + \norm{R_{\Phi}-R_{\Phi'}}_{C(X)} + \norm{Q-Q'}_{C(X)} } \leq \\
\leq \const\br{ \norm{\Phi-\Phi'}_{C^1(X)} + \norm{Q-Q'}_{C(X)} }.
\end{multline*}

By Theorem 7.16 from \cite{BBLZ-09} and by Lemma \ref{lem:1}, 
this implies that the inclusion $\DB \hookrightarrow \HF \br{L^2\br{X,g; \CC^{2}}}$ is graph continuous.
This completes the proof.


\section{Basic properties of the spectral flow}\label{sec:P0-P4}

There can be different versions of the definition of the spectral flow when one or both of the endpoints of the path is non-invertible. 
If a path is a loop up to a gauge transformation as in the first part of the paper, then the value of the spectral flow is independent of the choice of the definition. 
But for the proofs below
we have to fix some choice.

\begin{defi}
Let $\br{A_t}$ be an 1-parameter graph continuous family of closed self-adjoint Fredholm operators
in a separable complex Hilbert space $H$.
Take a small $\varepsilon>0$ such that $A_0$, $A_1$
have no eigenvalues in the interval $[-\varepsilon, 0)$. 
We put $\sp{A_t}$ be equal to $\sp{A_t+\varepsilon I}$,
were we use any of the (equivalent) definitions of the spectral flow for the path of operators with invertible endpoints from \cite{BLP-04, L-04}.
This definition 
does not depend on the choice of such $\varepsilon$.
\end{defi}

We will need the following properties of the spectral flow.

\medskip\noindent
\textbf{(P0) Zero crossing.}
In the absence of zero crossing the spectral flow vanishes.
More precisely, suppose $\gamma \colon [0,1] \to \DB$ is the continuous path such that 0 is not
the eigenvalue of $\gamma(t)$ for any $t\in [0,1]$. Then $\sp{\gamma} = 0$.

\medskip\noindent
\textbf{(P1) Homotopy invariance.}
The spectral flow along the continuous path $\gamma \colon [0,1] \to \DB$ does not change when
$\gamma$ changes continuously in the space of paths in $\DB$ with
the fixed endpoints (the same as the endpoints of $\gamma$).

In other words, for the continuous map $h \colon [0,1]\times [0,1] \to \DB$ such that
$h_s(0)\equiv \br{\Dt_0,B_0}$, $h_s(1)\equiv \br{\Dt_1,B_1}$,
we have $\sp{h_0(t)}\tin = \sp{h_1(t)}\tin$.

\medskip\noindent
\textbf{(P2) Path additivity.}
Suppose $\gamma \colon [a,c] \to \DB$ is a continuous path, $a\leq b\leq c$.
Then $\sp{\gamma(t)}_{t\in [a,c]} = \sp{\gamma(t)}_{t\in [a,b]} + \sp{\gamma(t)}_{t\in [b,c]}$.

\medskip\noindent
\textbf{(P3) Additivity with respect to direct sum.}
Let $N_1$, $N_2$ be natural numbers, $\br{\Dt_t^i,B_t^i}$ be continuous paths in $\DD_{N_i}\times\B_{N_i}$.
Then the spectral flow along the path $\br{\Dt_t^1\oplus\Dt_t^2, B_t^1\oplus B_t^2}$ is equal to the sum of the spectral flows along the paths $\br{\Dt_t^1,B_t^1}$ and $\br{\Dt_t^2,B_t^2}$.

\medskip\noindent
\textbf{(P4) Conjugacy invariance.} Let $J_{\pm}\colon X \to U(N)$ be unitary $N\times N$ matrices smoothly
dependent on $\x\in X$, $J = \matr{J_+ & 0 \\ 0 & J_-} \colon X \to U(2N)$, $\br{ \Dt_t, B_t }$ be a smooth path
in $\DB$. Then $\sp{ \Dt_t, B_t } = \sp{ J\Dt_t J\inv, J_- B_t J_-\inv }$.

More generally, if $H$ is a separable complex Hilbert space, $J$ is an unitary isomorphism of $H$,
$\br{A_t}$ is an 1-parameter graph continuous family of closed self-adjoint Fredholm operators,
then $\sp{A_t} = \sp{J A_t J\inv}$.

\begin{rem}
Properties (P1) and (P2) imply that the spectral flow along the path is opposite to the spectral flow along the same path passing in the opposite direction.
\end{rem}

\proof
By Lemmas \ref{lem:1}-\ref{lem:2}, the inclusion of $\DB$ into $\HF\br{L^2\br{X,g; \CC^{2N}}}$ is graph continuous.
So it is sufficient to prove Properties (P0-P4) for graph continuous paths in the space $\HF(H)$ 
of all closed self-adjoint Fredholm operators in separable Hilbert space $H$; 
this will imply properties (P0-P4) for the paths in $\DB$.

First three properties of the spectral flow for graph continuous paths in $\HF(H)$
are proved in \cite{BLP-04} (Proposition 2.2),
taking into account the convention from Section \ref{seq:sf} for the case when $\gamma(0)$ or $\gamma(1)$ are non-invertible.

Conjugacy invariance of the spectral flow for graph continuous paths in $\HF(H)$ follows from the uniqueness property of the spectral flow.
Namely, let $J$ be an unitary isomorphism of a separable complex Hilbert space $H$.
To each graph continuous path $\br{A_t}$ in $\HF(H)$  assign the integer $\spf\new(A_t) = \spf(J A_t J\inv)$. 
Then $\spf\new$ satisfies Concatenation, Homotopy and Normalization properties in the sense of \cite{L-04}.
By Theorem 5.9 from \cite{L-04}, this imply that $\spf\new$ equals $\spf$ for the paths in $\HF(H)$ with invertible endpoints. 
Taking into account our convention from Section \ref{seq:sf} and 
choosing a small $\varepsilon>0$ such that $A_0$, $A_1$ have no eigenvalues in the interval $[-\varepsilon, 0)$, 
we obtain
\[
\sp{J A_t J\inv} = \sp{J A_t J\inv +\varepsilon I}
 = \spf\new(A_t+\varepsilon I) = \sp{A_t+\varepsilon I} = \sp{A_t}.
\]

To prove (P3), consider graph continuous paths $(A_t)$, $(A'_t)$ in $\HF(H)$, $\HF(H')$ respectively. 
Suppose for a while that $A_0$, $A_1$, $A'_0$, and $A'_1$ are invertible.
The path $(A_t\oplus A'_t)\tin$ is homotopic to the concatenation of paths 
$(A_t\oplus A'_0)\tin$ and $(A_1\oplus A'_t)\tin$ in $\HF(H\oplus H')$. 
The spectral flow of the path $(A_t\oplus A'_0)$ in $\HF(H\oplus H')$ considered as the function of $(A_t)$ satisfies 
Concatenation, Homotopy and Normalization properties in the sense of \cite{L-04}, so
by the uniqueness property of the spectral flow from \cite{L-04} we have $\sp{A_t\oplus A'_0}=\sp{A_t}$. 
Similarly, $\sp{A_1\oplus A'_t}=\sp{A'_t}$. 
Therefore, $\sp{A_t\oplus A'_t} = \sp{A_t}+\sp{A'_t}$ for any paths $(A_t)$, $(A'_t)$ with invertible endpoints.
Taking into account our convention from Section \ref{seq:sf}, 
we obtain that $\sp{A_t\oplus A'_t} = \sp{A_t}+\sp{A'_t}$ for arbitrary paths $(A_t)$, $(A'_t)$.
This completes the proof.

\section{Independence of the choice of family $(Q_t)$}

Let us prove that the spectral flow along $\br{\Dt+Q_t, B}$ does not depend on the choice of $(Q_t)$ when $\Dt$, $B$, $\mu $ are fixed.

Let $Q_t$, $Q'_t$ be two 1-parameter families of smooth maps from $X$ to $\Herm(\CC^{2N})$
such that $Q_0=Q'_0=0$, $Q_1=Q'_1 = \mu \Dt \mu \inv - \Dt$.

The path $\Dt+Q_t$ can be continuously changed to the path $\Dt+Q'_t$ in the class of paths in $\DD$ with the fixed endpoints.
For example, we can take the homotopy $h(s,t) = \Dt + (1-s)Q_t + sQ'_t$.
By the homotopy invariance property (P1) of the spectral flow,
$\sp{\Dt+Q_t, B}\tin = \sp{\Dt+Q'_t, B}\tin$.

Therefore, if $Q_0=0$, $Q_1 = \mu \Dt \mu \inv - \Dt$ then 
$$\sp{\Dt+Q_t, B}\tin = F(X,g,N,\Dt,B,\mu )$$
for some integer-valued function $F$.
Now we will investigate the properties of this function.

\section{Independence of the choice of operator $\Dt$}
\label{sec:indep_D}

\textbf{1.}
Suppose that $\Dt_0$ is homotopic to $\Dt_1$ in $\DD$,
that is there exist a continuous 1-parameter family of the Dirac type operators $\Dt_s$ connecting $\Dt_0$ with $\Dt_1$.
We will show now that $F(X,g,N,\Dt_0,B,\mu ) = F(X,g,N,\Dt_1,B,\mu )$.

Let us consider 2-parameter family of the Dirac type operators $\Dt_{s,t} = (1-t)\Dt_s + t \mu \Dt_s \mu \inv$.
Note that $\Dt_{s,0} = \Dt_s$, $\Dt_{s,1} = \mu \Dt_s \mu \inv$,
$\Dt_{s,t} - \Dt_{s,0} = tQ_s$,
where $Q_s =  \mu \Dt_s \mu \inv - \Dt_s$ is the 1-parameter family of $2N\times 2N$ self-adjoint complex matrices smoothly dependent on $\x\in X$.

\begin{figure}[tbh]
\begin{center}
\includegraphics[width=0.4\textwidth]{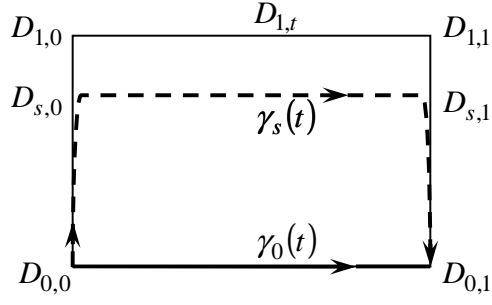}
\caption{Homotopy from $\gamma_0(t)$ to $\gamma_1(t)$}
\label{fig:Dst}
\end{center}
\end{figure}
Let us define the path $\gamma_1 \colon [0,3]\to\DD$ by the formula
\[
\gamma_1(t) =
\sys{ {ll}
\Dt_{t,0}, & t\in [0,1] \\
\Dt_{1,t-1}, & t\in [1,2] \\
\Dt_{3-t,1}, & t\in [2,3]
}
\]
In other words, we consequently go around the left, top and right sides of the rectangle on Fig.~\ref{fig:Dst} 
in clockwise direction. The path $\gamma_1$ can be continuously deformed to the path $\Dt_{0,t} = \Dt_0+tQ_0$
within the rectangle. For example, we can take as such a deformation the family
\[
\gamma_s(t) =
\sys{ {ll}
\Dt_{st,0}, & t\in [0,1] \\
\Dt_{s,t-1}, & t\in [1,2] \\
\Dt_{s(3-t),1}, & t\in [2,3]
}
\]
Then $\gamma_0(t)$ is the path $\br{\Dt_0+(t-1)Q_0}_{t\in [1,2]}$ concatenated with two steady paths, 
the spectral flows along which are zero by property (P0).

By the homotopic invariance property of the spectral flow,
\[
\spf\br{ \gamma_1(t), B}_{t\in [0,3]} =
\spf\br{ \gamma_0(t), B}_{t\in [0,3]} =
\sp{\Dt_0+tQ_0, B}\tin.
\]
On the other hand, the spectral flows along the first and the third parts of $\gamma_1$ are mutually reduced by (P4):
\[
\spf\br{ \gamma(t), B}_{t\in [0,1]} + \spf\br{ \gamma(t), B}_{t\in [2,3]} =
\sp{\Dt_s, B}_{s\in [0,1]} - \sp{\mu \Dt_s \mu \inv, B}_{s\in [0,1]} = 0.
\]
Therefore, $\sp{\Dt_0+tQ_0, B}\tin = \spf\br{ \gamma_1(t), B}_{t\in [1,2]} = \sp{\Dt_1+tQ_1, B}\tin$,
and $F(X,g,N,\Dt_0,B,\mu ) = F(X,g,N,\Dt_1,B,\mu )$.

\bigskip
\noindent
\textbf{2.}
Now we will simplify $\Dt$ step by step.

At first, we can continuously change 
$\Dt = \Dt_{\Phi,Q} = -i \br{\rho_1 \px{} + \rho_2 \py{} } + iR_{\Phi}(\x) + Q(\x)$ 
to the operator $\Dt_{\Phi, \: 0}$, for example, along the path
$\Dt_{\Phi,(1-s)Q}$.

Further, take a smooth map $h\colon [0,1]\times \GL^+(2,\R)\to \GL^+(2,\R)$ such that $h(0,\cdot)$ is the identity map, while the image of $h(1,\cdot)$ is the group $SO(2,\R)$ of $2\times 2$ orthogonal real matrices with determinant equal to one
(the existence of such a family is well known in any dimension; this is an easy application of Gram-Schmidt orthonormalisation procedure).
The operator $\Dt_{\Phi, \: 0}$ can be continuously changed in $\DD$ along the path $\Dt_{h(t,\Phi), \: 0}$ to the
operator $\Dt_{\Phi', \: 0}$, where
\[
\Phi'(\x) = \matr { \cos\varphi & \sin\varphi \\ -\sin\varphi & \cos\varphi } \in SO(2,\R)
\]
for some smooth function $\varphi$ from $X$ to $S^1$.
So we have 
$$F(X,g,N,\Dt_{\Phi,Q},B,\mu )=F(X,g,N,\Dt_{\Phi', \: 0},B,\mu ).$$

On the other hand, $\Dt_{\Phi', \: 0}$ can be represented as $J\inv
\Dt_{I, \: 0} J$, where $I=I_2$ is the identity $2\times 2$ matrix,
$$J(\x) = \matr {J_+ & 0 \\ 0 & J_-} = \matr {e^{i\varphi}I_N & 0 \\ 0 & I_N} \in U(2N).$$

Let $Q_t$ be an 1-parameter family of self-adjoint $2N\times 2N$ complex matrices such that
$Q_0=0$, $Q_1 = \mu  \Dt_{\Phi', \: 0} \mu \inv - \Dt_{\Phi', \: 0}$.
Applying property (P4) of the spectral flow, we obtain
\begin{multline*}
F(X,g,N,\Dt_{\Phi', \: 0},B,\mu ) = \sp{\Dt_{\Phi', \: 0}+Q_t, B} =
 \sp{ J \br{\Dt_{\Phi', \: 0} + Q_t} J\inv, J_- B J_-\inv} = \\
 = \sp{\Dt_{I, \: 0} + J Q_t J\inv, B} = F(X,g,N,\Dt_{I, \: 0},B,\mu ),
\end{multline*}
because $J Q_1 J\inv = \mu  \br{J\Dt_{\Phi', \: 0}J\inv} \mu \inv - J\Dt_{\Phi', \: 0}J\inv = \mu \Dt_{I, \: 0}\mu\inv - \Dt_{I, \: 0}$.

Therefore, $F(X,g,N,\Dt,B,\mu)$ does not depend on the choice of $\Dt\in\DD$, so from now on we will write $F(X,g,N,B,\mu)$ instead of $F(X,g,N,\Dt,B,\mu )$.

\section{Independence of the metric and invariance under the change of variables}
\label{sec:Xg}

We prove here that $F(X,g,N,B,\mu)$ is independent from the choice of the metric $g$ on $X$, 
invariant under the change of variables,
and does not depend on the geometry of $X$,
using the fact that the number of holes is the only topological invariant of the disk with holes,
and that the spectral flow is conjugacy invariant and does not depend on the choice of the operator.

Let $X$, $X'$ be compact planar domains, each bounded by $m$ smooth curves,
and $g$, $g'$ be Riemannian metrics on $X$, $X'$ respectively.

As well known, there exists an orientation-preserving diffeomorphism $f\colon X'\to X$.
\phys{In other words, there exist a smooth one-to-one change of variables 
$(x'^1,x'^2)=\x' \stackrel{f}{\rightarrow} \x=(x^1,x^2)$
with the smooth inverse and with positive Jacobian determinant $\det({\p\x}/{\p\x'})$, which transforms $X'$ onto $X$.}

We define $\theta$ as the smooth function from $X'$ to $\R^+$ such that $f^{\ast}\dvol = \theta\dvol'$,
where $\dvol$, $\dvol'$ are volume elements on $X$, $X'$ respectively.
\phys{As usual, by $f^{\ast}$ we denote the homomorphism from the differential forms (in particular, functions) on $X$
to the differential forms on $X'$, which is induced by $f$.
In coordinate form, $\theta(\x') = \frac{\sqrt{g(f(x'))}}{\sqrt{g'(x')}}\det(\frac{\p\x}{\p\x'})$.}

Diffeomorphism $f$ defines the unitary isomorphism $J$ from the Hilbert space $L^2\br{X,g; \CC^{2N}}$
to the Hilbert space $L^2\br{X',g'; \CC^{2N}}$,
$u \mapsto \sqrt{\theta}f^{\ast}u$. \phys{That is $(Ju)(\x') = \sqrt{\theta(\x')}u(f(\x'))$.}

Isomorphism $J$ transforms the operator $\Dt\in\DD\XgN$ with symbol $\rho$ 
to the symmetric operator $\Dt'=J\Dt J\inv$ on $X'$  with symbol $\rho'$.
For any $x'\in X'$, $x=f(x')$, any cotangent vector $\xi\in T^*_{x}X$, $\xi'=f^{\ast}\xi$, 
we have $\rho'(x',\xi') = \rho(x,\xi)$, 
that is $\rho'(\x') = \br{\frac{\p\x'}{\p\x}}\rho(x) = \br{\frac{\p\x'}{\p\x}}\Phi\sigma$ in coordinate representation.
The matrix $\br{\frac{\p\x'}{\p\x}}\Phi(x)$ is contained in $\GL^+(2,\R)$ for any $\x\in X$, so $\Dt'\in\DD_{X',\,g',\,N}$.

Let $\mu$ be a smooth function from $X$ to $U(1)$.
Taking the map $\mu'=f^{\ast}\mu$ from $X'$ to $U(1)$
and the map $B'=\left\| f^{\ast}\n \right\|_{g'}\inv f^{\ast}B$ from $\p X'$ to $\Herm(\CC^N)$,
we obtain
$$\mu'\Dt'\mu'^{-1} - \Dt' = \mu'\br{J\Dt J\inv}\mu'^{-1} - J\Dt J\inv = J\br{\mu\Dt\mu\inv - \Dt}J\inv.$$
So if $Q_t$ connects $Q_0=0$ with $Q_1=\mu\Dt\mu\inv - \Dt$, then
$Q'_t = JQ_t J\inv$ connects $Q'_0=0$ with $Q'_1=\mu'\Dt'\mu'^{-1} - \Dt'$,
and by the conjugacy invariance of the spectral flow (P4), we have
\[
\sp{\Dt+Q_t,B} = \spf\br{J(\Dt+Q_t)J\inv,B'} = \sp{\Dt'+Q'_t,B'}.
\]
However, $B'$ is homotopic to $f^{\ast}B$ in $\B_{X',\,N}$,
while the spectral flow of $(\Dt'+Q'_t, \widetilde{B})$ is invariant under the continuous change of $\widetilde{B}$ in $\B_{X',\,N}$ (this is verified in a way similar to the proof in Section \ref{sec:indep_D}).
Therefore $\sp{\Dt'+Q'_t,B'} = \sp{\Dt'+Q'_t, f^{\ast}B}$,
and finally we obtain
\begin{equation}
    F(X,g,N,B,\mu) = F(X',g',N,f^{\ast}B,f^{\ast}\mu).
\end{equation}
\noindent
This completes the proof.

In particular, for any two metrics $g$, $g'$ on the same $X$, using the identity diffeomorphism $f$, we have
$$F(X,g,N,B,\mu) = F(X,g',N,B,\mu).$$

Further we will write $F(N,B,\mu)$ instead of $F(X,g,N,B,\mu)$.

\section{Boundary conditions}

Let us investigate the dependence of $F(N,B,\mu )$ on $B$.

$F(N,B,\mu )$ does not change when $B$ continuously changes in
$\B$; this is verified in a way similar to the proof in Section
\ref{sec:indep_D}.

Let $b_j$ be the number of negative eigenvalues of $B$ (counting multiplicities) on $\p X_j$.
We prove that the ordered set $\hat{b} = \br{b_j}_{j=1}^m$
uniquely determines $B$ up to continuous variation of $B$ in $\B$.

Obviously, $\hat{b}$ is invariant with respect to such variations,
so we only have to prove that any two $B$, $B'$ with the same $\hat{b}$ are homotopic.
It is sufficient to prove that any smooth map $A$ from the circle $S^1$ to the space of complex self-adjoint invertible $N\times N$ matrices is homotopic (in the space of all such maps with $C^1$-metric)
to the steady map sending $S^1$ to the point $(-I_k)\oplus I_{N-k} \in \Herm(\CC^N)$,
where $k$ is the number of negative eigenvalues of $A(\x)$, $\x\in S^1$.

\bigskip
\noindent
\textbf{1.} Let us consider the continuous 1-parameter family
$A_s$ of smooth maps from $S^1$ to the space of complex self-adjoint invertible $N\times N$ matrices 
defined by the formula $A_s = A\cdot\br{ (1-s)I_N +sA^2 }^{-1/2}$. 
This expression is correct because $(1-s)I_N +sA^2$ is self-adjoint and positive definite for any $s\in [0,1]$. 
The family $A_s$ gives us the deformation from $A=A_0$ to smooth map
$A_1$ from $S^1$ to the space of self-adjoint \textit{unitary} $N\times N$ matrices.

\bigskip
\noindent
\textbf{2.}
The connected component of $A_1(\x)$ in the space of self-adjoint unitary $N\times N$ matrices is diffeomorphic to the space $Gr_{\CC}(k,N)$ of all $k$-dimensional linear subspaces of $\CC^N$.
This diffeomorphism is defined by the correspondence $U\mapsto \Ker\br{I_N+U}$, which associates with $U$ the invariant subspace $V\subseteq\CC^N$ of U corresponding to eigenvalue $-1$ of $U$.
The inverse diffeomorphism is defined by the formula $V\mapsto U=(-I)_V \oplus I_{V^{\bot}}$.

The complex Grassmanian $Gr_{\CC}(k,N)$ is known to be simply connected,
so any two continuous maps from the circle to $Gr_{\CC}(k,N)$ are homotopic.
Taking into account that $Gr_{\CC}(k,N)$ is the smooth manifold,
we obtain that the space of smooth maps from the circle to $Gr_{\CC}(k,N)$ (with $C^1$-metric) is path-connected.
The same is true for the connected component of the space of self-adjoint unitary $N\times N$ matrices 
which is diffeomorphic to $Gr_{\CC}(k,N)$,
so $A$ can be continuously changed in the class of smooth maps to the steady map $\x \mapsto (-I_k)\oplus I_{N-k}$.
This completes the proof.

\section{Gauge transformations}\label{seq:mu}

\textbf{1.}
We will prove that $F$ is linear in $\mu$, that is
$F(N,B,\mu_1 \mu_2) = F(N,B,\mu_1)+F(N,B,\mu_2)$ for any smooth functions $\mu_1, \mu_2 \colon X \to U(1)$.

Let $Q_i = \mu_i \Dt \mu_i\inv - \Dt$.
Then $Q_1+Q_2 = \br{\mu_1 \mu_2}\Dt\br{\mu_1 \mu_2}\inv - \Dt$,
so by definition $F(N,B,\mu_1 \mu_2)$ is equal to the spectral flow along the path $\br{ \Dt+P_t, B }_{t\in [0,2]}$, where $P_0=0$, $P_2=Q_1+Q_2$.
We can take $P_t$ composed from two parts: from $0$ to $Q_1$ and then from $Q_1$ to $Q_1+Q_2$, for example,
\[
P_t =
\sys{ {ll}
tQ_1, & t\in [0,1] \\
Q_1+(t-1)Q_2, & t\in [1,2]
}
\]
Using the property (P2) of the spectral flow, we obtain
\begin{multline*}
F(N,B,\mu_1 \mu_2) =
\sp{ \Dt+P_t, B }_{t\in [0,1]} + \sp{ \Dt +P_t, B }_{t\in [1,2]} = \\
= \sp{ \Dt+tQ_1, B }\tin + \spf\br{ (\Dt+Q_1)+tQ_2, B }\tin = \\
= F(N,B,\mu_1)+F(N,B,\mu_2),
\end{multline*}
so $F$ is linear in $\mu$.

\bigskip
\noindent
\textbf{2.}
By $M$ denote the set of equivalence classes of
smooth functions $\mu\colon X\to U(1)$, where two functions are
equivalent if one of them can be continuously changed to another
in the space of smooth functions from $X$ to $U(1)$ (with $C^1$-metric).
We will consider $M$ as the Abelian group, where the group structure on
$M$ is induced by the group structure on $U(1)$.
It is well known that
\[
M = \set{ \br{\mu_1,..,\mu_m}\in\Z^m\colon \sum\mu_j = 0 },
\]
with the group structure induced from $\Z^m$,
and the class of $\mu$ in $M$ is defined by the $m$-tuple $\hat{\mu} = \br{\mu_j}$,
where $\mu_j$ is the degree of the restriction of $\mu $ to $\p X_j$.

Let us prove that $F(N,B,\mu)$ depends only on the class of $\mu$ in $M$.

Suppose that $\mu_t$ is a continuous path in  the space of smooth functions from $X$ to $U(1)$
such that $\mu_0(\x)\equiv 1$.
By the previous clause, it is sufficient to prove that $F(N,B,\mu_1)=0$.
Let us take $Q_t = \mu_t \Dt {\mu_t}\inv - \Dt$.
Taking into account that $Q_1 =  \mu_1  \Dt \mu_1\inv - \Dt$, we obtain
$F(N,B,\mu_1) = \sp{ \Dt +Q_t, B }$.
But all the operators $\brop{\Dt +Q_t,B}$ are conjugate to $\DtB$ by $\mu_t$ and therefore are isospectral.
Let $\varepsilon>0$ be such that $\DtB$ has no zero eigenvalues in the interval $[-\varepsilon,0)$.
Then $\sp{\Dt +Q_t, B} = \sp{\Dt +Q_t +\varepsilon I_{2N}, B} = 0$ by (P0) because
all the operators $\br{\Dt +Q_t +\varepsilon I_{2N}, B}$ have no zero eigenvalues.
This completes the proof.

\section{Bilinearity}

In the previous sections we have proven that $F$ depends only on the integer numbers 
$N$, $b_1, \ldots, b_m, \mu_1, \ldots, \mu_m$.
Now we will study this dependence more closely.

By $S$ denote the set of all possible $(m+1)$-tuples $\br{N, b_1,\ldots,b_m}$:
\[
S = \set{ \br{N, b_1,\ldots,b_m}\in\Z^{m+1}\colon N\geq 1, 0\leq b_j\leq N }.
\]
$F$ defines the map from $S\times M$ to $\Z$ 
(which we denote by the same letter $F$ for simplicity)
satisfying the following conditions:
\begin{align*}
      F\br{N,\hat{b},\hat{\mu}\oplus\hat{\mu}'} &= F\br{N,\hat{b},\hat{\mu}} + F\br{N,\hat{b},\hat{\mu}'} \\
      F\br{N+N',\hat{b}\oplus\hat{b}',\hat{\mu}} &= F\br{N,\hat{b},\hat{\mu}} + F\br{N',\hat{b}',\hat{\mu }}
\end{align*}
where $\hat{\mu}=\br{\mu_j}_{j=1\ldots m}$, $\hat{b}=\br{b_j}_{j=1\ldots m}$,
symbol $\oplus$ denotes the componentwise addition.
Indeed, the first equality has been proven in section \ref{seq:mu}, while the
second equality is by the property (P3) of the spectral flow.

Hence $F$ is a bilinear function, and therefore there is a homomorphism from
$\Z^{m+1}\otimes M$ to $\Z$ such that $F$ can be represented as the composition
\begin{equation}\label{eq:comp}
S\times M \hookrightarrow \Z^{m+1}\times M \to \Z^{m+1}\otimes M \rightarrow \Z,
\end{equation}
where the first arrow is induced by the natural embedding of $S$
into $\Z^{m+1}$, and the second arrow is the canonical map of the direct product to the tensor product.

\medskip
Let us consider operator \eqref{eq:Qt} with boundary condition \eqref{eq:bc}.
If $\br{\Dir+Q_t}u=0$ and $i\br{\n_1 + i\n_2} u^+ = B u^-$ on $\p X$, then
\begin{multline*}
\int_{\p X} \bra{ B(\x)u^-, u^- }\d s =
\int_{\p X} \bra{ i\br{\n_1+i\n_2}u^+, u^- }\d s = \\
 = \int_X \bra{ \br{-i\br{\p_1+i\p_2}+q_t}u^+, u^- }\d x^1 \d x^2 -
    \int_X \bra{ u^+, \br{-i\br{\p_1-i\p_2}+\overline{q}_t}u^- }\d x^1 \d x^2 = 0,
\end{multline*}
where $\d s$ is the length element on $\p X$.

Suppose now that the sign of $B$ is the same on all boundary components.
Then from the last equality we have ${u^-}\equiv 0$ on $\p X$, $u^+ = -i\br{\n_1-i\n_2}Bu^- \equiv 0$ on $\p X$.
Thus $u\equiv 0$ on $X$ by the weak inner unique continuation property of Dirac operator \cite{BBW-93}.
So $\br{\Dir+Q_t,B}$ has no zero eigenvalues at any $t$ for such B, and by Property (P0) $\sp{\Dir+Q_t,B}=0$.
Finally we obtain $F(1,\hat{0},\hat{\mu}) = F(1,\hat{1},\hat{\mu}) = 0$ at any $\hat{\mu}$,
where we denote $\hat{0} = \br{0,\ldots, \: 0}, \; \hat{1} = \br{1,\ldots,1} \in \Z^m$.

Let us consider the group $M'$ which is quotient of $\Z^{m+1}$ by subgroup spanned by elements 
$\br{1,\hat{0}}, \br{1,\hat{1}} \in \Z^{m+1}$. 
Note that $M'$ coincides with the quotient group ${\Z^{m}}/\bra{\hat{1}}$, so it is naturally
isomorphic to the Abelian group $\Hom(M,\Z)$ of all homomorphisms of $M$ to $\Z$.

By previous arguments, there exists such homomorphism
$\widetilde{F}\colon M'\otimes M \to \Z$ that $F$ is the composition
of the following homomorphisms:
\begin{equation}
  S\times M \hookrightarrow \Z^{m+1}\times M \to \Z^{m+1}\otimes M \to M'\otimes M \stackrel{\widetilde{F}}{\rightarrow} \Z,
\end{equation}
where the first two arrows are the same as in \eqref{eq:comp},
and the third arrow is induced by the natural projection $\Z^{m+1}\to M'$.

\section{Invariance under the action of symmetric group}

Let $\Diff^+(X)$ be the group of all diffeomorphisms of $X$ preserving orientation, $f\in\Diff^+(X)$.
As it was shown in Section \ref{sec:Xg}, $F\br{N,f^{\ast}B, f^{\ast}\mu} = F(N,B,\mu)$,
and hence
$$F\br{N,f^{\ast}\hat{b}, f^{\ast}\hat{\mu}} = F(N,\hat{b},\hat{\mu}),$$
where $f^{\ast}$ acts on $\hat{b}$ and $\hat{\mu}$ by the permutation of the coordinates,
corresponding to the permutation of the boundary components of $X$ by $f$.
It is well known that any permutation of the boundary components of $X$ is realized by some element of $\Diff^+(X)$.
Thus $F(N,\hat{b},\hat{\mu})$ is invariant under the action of
symmetric group $S_m$ (the group of permutations of $m$ elements) on $\br{\hat{b}, \hat{\mu}}$ by the permutations of the coordinates.

All permutations of the coordinates leave he element $\hat{1}$ of $\Z^m$ invariant,
so $S_m$ acts on $M' = {\Z^{m}}/\bra{\hat{1}}$ in exactly the same way,
and $\widetilde{F}$ is invariant under the action of $S_m$, too.

Extending $\widetilde{F}$ by linearity from $M'\otimes M$ to $V'\otimes V$,
$V' = M'\otimes \CC = {\CC^m}/{\bra{\hat{1}}}$,
$V = M\otimes \CC = \set{v\in\CC^m \colon \sum v_j = 0}$,
we obtain homomorphism $\widetilde{F}_{\CC}\colon V'\otimes V \to \CC$,
coinciding with $\widetilde{F}$ on the lattice $M'\otimes M \subset V'\otimes V$.
Obviously, $\widetilde{F}_{\CC}$ is invariant with respect to the action of $S_m$ on $V'\otimes V$ as well.

$V'$ and $\Hom_{\CC}(V,\CC)$ coincide as the representations of $S_m$,
so the vector space of all invariant homomorphisms from $V'\otimes V$ to $\CC$ is isomorphic to the
vector space of all equivariant homomorphisms $V \to V$.
But the latter space is 1-dimensional by Schur's lemma, because $V$ is the irreducible representation of $S_m$ \cite{Fulton}.
So $\widetilde{F}_{\CC}\br{v'\otimes v} = c\sum_j v'_j v_j$ for some constant $c\in\CC$,
and $F(N,\hat{b},\hat{\mu}) = c\sum b_j \mu_j$, where $c$ depends only on $m$.

On the other hand, $F$ is integer-valued and, in particular, $c=F\br{1,(0,1),(-1,1)}\in\Z$.

Finally, we obtain $\sp{\Dt+Q_t, B}\tin = c_m \sumj b_j \mu_j$,
where $c_m$ is the integer constant depending on $m$ only, and Theorem \ref{thm:gDirac} is proved.

\section*{Acknowledgements}
\addtocontents{toc}{\vspace{5mm}}
\addcontsec{Acknowledgements}

The author is very grateful to M.I.~Katsnelson for attracting author's attention to this circle of questions and for explaining that these problems are of potential importance for the condensed matter physics. The latter served as an important motivation for this work. The conversations with M.I.~Katsnelson helped to understand what kind of difficulties a physicist may have while reading a mathematical paper. Hopefully, this led to a more physicist-friendly style of this paper. M.I. Katsnelson also informed the author about the papers \cite{AhmBeen, Berry}. 
The main part of the text was written during author's visit to Radboud University Nijmegen at the invitation of M.I.~Katsnelson and with the financial support from the Stichting voor Fundamenteel Onderzoek der Materie (FOM).

The author thanks M.~Lesch and F.V.~Petrov for their help with subtle questions of the operator theory and the theory of Sobolev spaces.
The author is very grateful to N.V.~Ivanov for his continuous support and his efforts to improve the exposition in this paper. 
The author is also thankful to M.~Braverman, A.~Gorokhovsky, M.E.~Kazaryan, and I.A.~Panin for stimulating discussions of various topics related to this paper.

This work was partially supported by the RFBR grant 09-01-00139-a (Russia),
and by the Program for Basic Research of Mathematical Sciences Branch of Russian Academy of Sciences (project 12-T-1-1003).
It was partially done during author's stay at Max Planck Institute for Mathematics (Bonn, Germany);
the author is  grateful to this institution for the hospitality and the excellent working conditions.

\bigskip
\bigskip
\noindent
\textit{Ural Federal University, Ekaterinburg, Russia}  
\\
\textit{pmf@imm.uran.ru}


\begin{thebibliography}{99}
\addcontsec{References}

\bibitem{AhmBeen}
A. R. Akhmerov and C. W. J. Beenakker.
Boundary conditions for Dirac fermions on a terminated honeycomb lattice.
Phys. Rev. B 77 (2008) 085423;
arXiv:0710.2723v3 [cond-mat.mes-hall]

\bibitem{APS-76}
M. F. Atiyah, V. K. Patodi, and I. M. Singer.
Spectral asymmetry and Riemannian geometry. III.
Math. Proceedings of the Cambridge Philosophical Society 79, part 1 (1976) 71.

\bibitem{BBW-93}
 B. Booss-Bavnbek, K. P. Wojciechhowski.
 Elliptic Boundary Problems for Dirac Operators,
 Birkhauser, 1993.

\bibitem{BLP-04}
 B. Booss-Bavnbek, M. Lesch, and J. Phillips.
 Unbounded Fredholm Operators and Spectral Flow.
 Canad. J. Math. 57 (2005) 225;
 arXiv:math/0108014v3 [math.FA]

\bibitem{L-04}
M. Lesch.
The uniqueness of the spectral flow on spaces of unbounded self-adjoint Fredholm operators.
In: Spectral geometry of manifolds with boundary and decomposition of manifolds (B.~Booss-Bavnbek, G.~Grubb, and K.P. Wojciechowski, eds.), Cont. Math., vol. 366, Amer. Math. Soc., 2005, pp. 193-224;
arXiv:math.FA/0401411

\bibitem{BLP-02}
B. Booss-Bavnbek, M. Lesch, and J. Phillips.
Spectral Flow of Paths of Self-Adjoint Fredholm Operators.
Nucl. Phys. B Proceedings Supplement 104 (2002), 177.

\bibitem{BL-01}
J. Bruning and M. Lesch.
On boundary value problems for Dirac type operators: I. Regularity and self-adjointness.
Funct. Anal. 185 (2001) 1;
arXiv:math/9905181v2 [math.FA]

\bibitem{BBLZ-09}
B. Booss-Bavnbek, M. Lesch, and C. Zhu.
The Calderon Projection: New Definition and Applications.
J. of Geometry and Physics 59 (2009), No. 7, 784;
arXiv:0803.4160v3 [math.DG]

\bibitem{Berry}
M. V. Berry and R. J. Mondragon.
Neutrino billiards: time-reversal symmetry-breaking without magnetic fields.
Proc. R. Soc. London, Ser. A 412 (1987) 53.

\bibitem{Prokh}
M.F. Prokhorova. The spectral flow for first order elliptic operators on a compact surface. In preparation.

\bibitem{Str-67}
R. S. Strichartz.
Multipliers on fractional Sobolev spaces.
J. of Mathematics and Mechanics, 1967. Vol. 16, No. 9, pp. 1031-1060.

\bibitem{Fulton}
W. Fulton, J. Harris. Representation theory. A first course. Graduate Texts in Mathematics, Readings in Mathematics, 129. Springer-Verlag, 1991.


\end{thebibliography}
\end{document}